\documentclass[conference]{IEEEtran}

\usepackage[table]{xcolor} %
\usepackage{array}
\usepackage{booktabs}
\usepackage{multirow}
\usepackage{tcolorbox}
\usepackage{enumitem}
\usepackage{float}
\usepackage{amsmath}
\usepackage{amssymb}
\usepackage{graphicx}
\usepackage{pifont}
\usepackage{mdframed}
\usepackage{mathtools}
\usepackage{algorithm}
\usepackage{algpseudocode}
\usepackage{tabularx}
\usepackage{makecell}
\usepackage{amsthm}
\usepackage{diagbox}
\usepackage{threeparttable}
\usepackage{caption}
\usepackage{subcaption}
\usepackage{fontawesome5}
\usepackage{tikz}
\usepackage{url}
\usepackage{hyperref}
\usepackage{cite}
\renewcommand{\footnoterule}{%
    \kern -3pt
    \hrule width 0.35\columnwidth height 0.4pt
    \kern 2.6pt
}
\usepackage[available,functional,reproduced]{ndssbadges}

\theoremstyle{definition}
\newtheorem{definition}{\textbf{Definition}}[section]
\theoremstyle{plain}
\newtheorem{theorem}{Theorem}[section]

\newtheorem{corollary}[theorem]{Corollary}

\newcommand{\cmark}{\ding{51}}  %
\newcommand{\xmark}{\ding{55}}  %
\newcommand{\hmark}{\ensuremath{\approx}\kern-0.15em\ding{51}} %
\newcommand{\mingxun}[1]{\textcolor{red}{MZ: #1}}
\newcommand{\ignore}[1]{}
\newcommand{\name}{OrionLink}
\newcommand{\negl}{\mathsf{negl}}
\newcommand{\re}[1]{{\color{black}#1}} %
\newcommand{\reon}{\color{black}} %
\newcommand{\ret}[1]{{\color{black}#1}} %

\definecolor{buaa-blue}{rgb}{0.118, 0.220, 0.420}
\definecolor{buaa-orange}{rgb}{0.77, 0.35, 0.06}
\definecolor{buaa-green}{rgb}{0.23, 0.37, 0.13}
\definecolor{note-blue}{rgb}{0.239, 0.565, 0.859}
\definecolor{deepred}{rgb}{0.5216, 0.0745, 0.1294}

\newtcolorbox{takeawaybox}{
  colframe=black!30,
  colback=black!5,
  boxrule=0.5pt,
  arc=1pt,
  left=4pt,
  right=4pt,
  top=2pt,
  bottom=2pt
}

\newcommand{\circnum}[1]{%
  \tikz[baseline=(n.base)]{%
    \node[draw,circle,minimum size=1.0em,inner sep=0.5pt](n){#1};}}

\newcommand{\circnumtwo}[1]{%
  \tikz[baseline=(n.base)]{%
    \node[draw,circle,minimum size=1.0em,inner sep=0.5pt](n){\small #1};%
  }%
}

\newcommand{\narrowxleftarrow}[2][]{%
  \mathrel{%
    \xleftarrow[#1]{\smash{\raisebox{-0.4ex}{$\scriptscriptstyle #2$}}}%
  }%
}

\begin{document}

\title{OrionLink: Single Sign-On with\\Oblivious Identity Brokers}

\author{
\IEEEauthorblockN{Yijia Fang\IEEEauthorrefmark{1}\IEEEauthorrefmark{3},
Bingyu Li\IEEEauthorrefmark{1}\textsuperscript{\ding{41}},
Luchao Jin\IEEEauthorrefmark{1},
Jingqiang Lin\IEEEauthorrefmark{2},\\
Mingxun Zhou\IEEEauthorrefmark{3},
Zhijintong Zhang\IEEEauthorrefmark{1},
Linghui Li\IEEEauthorrefmark{4},
Qianhong Wu\IEEEauthorrefmark{1}}
\IEEEauthorblockA{%
\begin{tabular}{@{}c@{\hspace{1.5em}}c@{}}
\IEEEauthorrefmark{1}School of Cyber Science and Technology, &
\IEEEauthorrefmark{2}School of Cyber Science and Technology,\\
Beihang University, China &
University of Science and Technology of China, China\\
\IEEEauthorrefmark{3}The Hong Kong University of Science and Technology &
\IEEEauthorrefmark{4}Beijing University of Posts and Telecommunications
\end{tabular}\\
\{fangyj, libingyu, jinluchao, zhangzjt, qianhong.wu\}@buaa.edu.cn\\
linjq@ustc.edu.cn, mingxunz@ust.hk, lilinghui@bupt.edu.cn}}

\IEEEoverridecommandlockouts
\makeatletter\def\@IEEEpubidpullup{3.5\baselineskip}\makeatother
\IEEEpubid{\parbox{\columnwidth}{
    Network and Distributed System Security (NDSS) Symposium 2027\\
    22--26 March 2027, Seoul, Republic of Korea\\
    ISBN 978-1-970672-09-1\\
    https://dx.doi.org/10.14722/ndss.2027.231871\\
    www.ndss-symposium.org
}
\hspace{\columnsep}\makebox[\columnwidth]{}}

\maketitle
\thispagestyle{IEEEtitlepagestyle}

\begingroup
\renewcommand{\thefootnote}{}
\footnotetext{\textsuperscript{\ding{41}} Bingyu Li is the corresponding author.}
\endgroup

\begin{abstract}
Single Sign-On (SSO) is a federated authentication mechanism that improves user experience by enabling access to multiple services, or Relying Parties (RPs), through a single authentication at an Identity Provider (IdP).
To reduce integration complexity, many RPs rely on third-party \emph{Identity Brokers (Brokers)} that aggregate numerous IdPs into a unified interface, thereby simplifying protocol management and interoperability.
However, while brokers ease development, they introduce a fundamental security and privacy risk: sensitive authentication credentials and user behavioral data are exposed to the broker, creating a high-value single point of failure.
Existing SSO solutions provide limited protection against these risks.

With the goal of achieving broker obliviousness, we address these challenges at their root and present \emph{OrionLink}, a comprehensive privacy-enhancing SSO pipeline that is compatible with existing brokered SSO architectures and systematically mitigates broker-induced risks without sacrificing usability.
OrionLink incorporates two core components:
(\emph{a}) an \textit{anonymous pseudonym derivation scheme} that generates unlinkable, per-RP identifiers to prevent cross-service user tracking;
(\emph{b}) an \textit{anonymous key exchange protocol} that establishes end-to-end encrypted channels to protect user data from broker visibility.
We further instantiate OrionLink with two new features, preserving broker functionality and deployability. We implemented a prototype of OrionLink within a brokered SSO environment, and experimental evaluation shows that the framework achieves substantial security and privacy improvements with acceptable performance overhead.
\end{abstract}


\IEEEpeerreviewmaketitle

\section{Introduction}
Single Sign-On (SSO) is one of the most widely adopted federated authentication mechanisms, commonly realized through standards such as OpenID Connect 1.0 (OIDC) \cite{sakimura2014openid} and OAuth 2.0 \cite{hardt2012oauth}. 
It allows users to access websites and services, known as relying parties (RPs), with a pre-registered account managed by an identity provider (IdP), 
significantly reducing registration and account management burdens. 
The adoption of SSO is widespread; more than 45,000 websites in the Tranco Top 1M \cite{pochat2018tranco} rely on major IdPs such as Apple, Google, Microsoft, and Facebook \cite{jannett2024sok}. 
Over the past decade, numerous studies have analyzed the security of SSO \cite{ghasemisharif2022towards,fett2016comprehensive,ghasemisharif2018single,fett2014expressive,jannett2024sok,liu2021investigation,philippaerts2022oauch}, 
and the IETF regularly revises its specifications \cite{ietf-oauth-v2-1-13,sakimura2014openid}, 
such as the \textit{Security Best Current Practices (SBCPs)}~\cite{lodderstedt2025rfc}, to incorporate these findings.

\ignore{
\mingxun{Remove the following paragraph. Not necessary. Jump right into Brokered SSO.}
\textbf{Challenges of SSO Development.} 
Despite standardization, the complexity of authentication protocols and their rapid evolution make correct deployment difficult for RPs. 
Measurement studies consistently reveal widespread noncompliance with recommended practices \cite{shi2019mossot,drakonakis2020cookie,wei2021moscan,jannett2022distinct,innocenti2023oauth,westers2024single,rahat2022cerberus}. 
For example, data from July 2023 showed that only 25\% of RPs (10,527 of 45,532) complied with SBCPs \cite{jannett2024sok}. 
Moreover, differences across IdPs force RP developers to repeatedly adapt implementations, increasing overhead and introducing security pitfalls.
}

\noindent \textbf{Brokered SSO.}
Brokered SSO is an increasingly adopted architecture in which an intermediate \emph{identity broker} (IdB) mediates interactions between relying parties (RPs) and identity providers (IdPs). 
The broker exposes a unified authentication interface to RPs 
while internally managing the heterogeneous, 
protocol-specific authentication workflows of multiple IdPs. 
This design relieves RPs from the need to directly integrate with diverse identity systems.
Figure~\ref{fig:three-party sso} contrasts conventional three-party SSO with brokered SSO.

The primary advantage of brokered SSO is reduced deployment and maintenance overhead. 
By integrating with a single broker, an RP can simultaneously support authentication from many IdPs, analogous to payment aggregators (e.g., Stripe or Alipay) abstracting bank-specific protocols for online merchants. 

Adoption of brokered SSO has grown rapidly. 
A 2025 measurement study \cite{innocenti2025only} reports that approximately 25\% of SSO-enabled websites rely on identity brokers. 
Emerging standards, such as the Model Context Protocol (MCP) \cite{hou2025model}, further increase demand for brokered authentication \cite{mcpauth, cloudflare}. 
As a result, identity brokers have become a critical component of the authentication ecosystem, 
with significant security and privacy implications for both RPs and end users (see Appendix~\ref{app:Measurements of Ecosystem}).

\ignore{

To mitigate these development burdens and reduce security risks from misconfigurations, a new architectural pattern has emerged: 
delegating RP-IdP interactions to third-party services known as \emph{Identity Brokers (Brokers)}. 
Brokers offer three benefits: (\emph{a}) \emph{compliance to standards}, by offloading protocol implementation to specialized providers; 
(\emph{b}) \emph{unified integration}, enabling support for diverse IdPs with a single integration; 
and (\emph{c}) \emph{data management}, by centralizing token verification, mapping, re-issuance, and lifecycle management. 
The traditional three-party SSO flow involves direct communication between the user, IdP, and RP. 
In brokered SSO, the broker acts as an intermediary: 
the RP communicates only with the broker, which in turn manages interactions with the appropriate IdP. 
This distinction is illustrated in Figure \ref{fig:three-party sso}.
}

\begin{figure}[t]
  \centering
  \includegraphics[width=\linewidth]{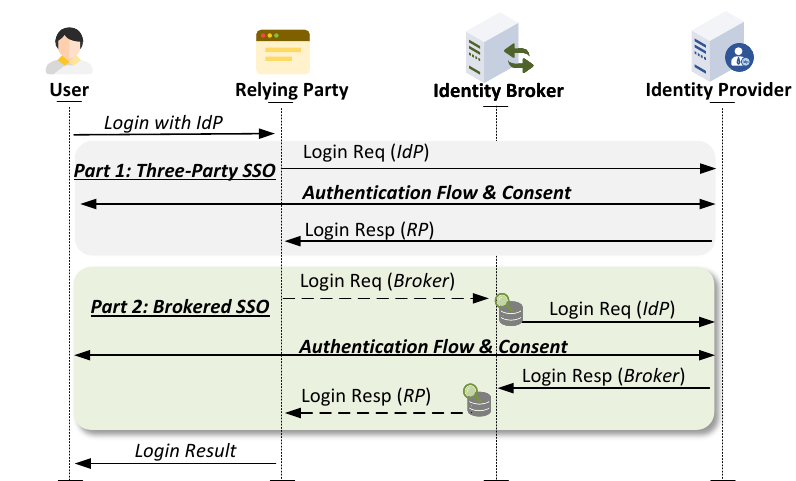}
  \caption{Traditional three-party SSO (top) features direct RP–IdP interaction, while brokered SSO (bottom) inserts an identity broker that relays and re-issues authentication tokens.}

  \label{fig:three-party sso}
\end{figure}

\noindent \textbf{Data Exposure Risk in Brokered SSO.}
While identity brokers simplify SSO deployment, they introduce a fundamental security and privacy risk
by concentrating sensitive authentication and behavioral data at a single intermediary, 
thereby creating a high-value point of failure. 
At a minimum, brokers can observe users’ cross-site authentication activities, enabling fine-grained tracking and potential secondary use of privacy-sensitive information.

More critically, brokers routinely process high-impact authentication data across many RPs, 
including plaintext user identifiers, access tokens, and session cookies. 
This aggregation makes brokers attractive targets for attackers, with potentially severe downstream consequences. 
In 2023, for example, a major identity broker, \emph{Okta}, was compromised via its customer support systems, exposing HTTP Archive (HAR) files that contained authentication tokens and session cookies for multiple RPs~\cite{OktaDataBreach}. 
The incident affected several foundational infrastructure providers, including Cloudflare \cite{cloudflare} and 1Password \cite{password}, and propagated to their dependent services.

Beyond isolated breaches, systemic risks are also evident. 
A large-scale study by Innocenti et al.~\cite{innocenti2025only} found that 49 brokers inadequately validate redirection chains, potentially exposing user data or authentication tokens to attackers.
In some cases, these flaws enable full account takeover through token leakage.

Given the current state of brokered SSO deployments, we pose the following fundamental research question: 

\begin{itemize}[leftmargin=6mm]
\item[]
\textit{
Is it possible to design a brokered SSO system that minimizes data leakage at the broker while still providing a unified authentication interface with practical efficiency?}
\end{itemize}

\noindent\textbf{Our Contribution.}
We take a more fundamental approach to addressing the security and privacy risks of brokered SSO.  
Rather than introducing additional protection mechanisms to secure data already exposed to identity brokers, we aim to \emph{minimize the information observable by the broker in the first place}.  
To this end, we design a new brokered SSO pipeline, \textbf{\name}, 
based on cryptographic engineering approaches.

The central design goal of \name{} is \emph{broker obliviousness}. 
In \name{}, the broker does not learn users’ real identities and acts solely as a forwarding intermediary between RPs and IdPs. 
Sensitive authentication data are protected using end-to-end encryption (E2EE), ensuring that even a compromised broker cannot access critical private information or exfiltrate it retrospectively (as shown in Figure \ref{fig:improving-privacy}).

Concretely, \name{} enforces multiple forms of broker obliviousness through carefully tailored cryptographic designs  (summarized in Table~\ref{tab:broker-privacy}), including:

\begin{enumerate}[itemsep=0pt,leftmargin=12pt]
    \item \emph{Identity Obliviousness:
    The broker does not learn the user’s real identity in any SSO session.} 
    In each session, the RP derives an anonymous identifier and generates a corresponding authenticity proof that is verified by the IdP. 
    The IdP then applies an Oblivious Pseudo-Random Function (OPRF)~\cite{casacuberta2022sok} to blindly derive a pairwise pseudonym for the user, bound to the RP’s anonymous identifier. 
    This construction ensures that pseudonyms are unique per RP, preventing cross-domain tracking while preserving the session consistency required for broker-side data handling.  
    
    \item \emph{Authorization Credential Obliviousness:}
    \textit{The broker does not learn authorization credentials (e.g., access tokens) used to access user resources at the IdP.} 
    In each session, the IdP anonymously authenticates the RP and verifies that it is registered and authorized by the user. 
    We further design an \emph{anonymous authentication key agreement (AAKA)} protocol that enables the IdP and RP to establish a secure session key.
    This key is used to encrypt authorization credentials in the login response, 
    which is forwarded by the broker and can only be decrypted by the authorized RP within the active session. 
    As a result, the broker never observes authorization credentials in plaintext, mitigating the risk of credential leakage and unauthorized resource access. 
    
    \item \emph{Associated Personal Information Obliviousness:}
   \textit{ The broker does not learn the user's personal information (e.g., email address, name, or phone number).}
   Similar to authorization credentials, personal information is encrypted by the IdP in the login response and forwarded through the broker such that only the authorized RP can decrypt it. 
   Since the broker neither possesses authorization credentials nor access to plaintext user attributes, 
   it cannot harvest personal information or query user resources at the IdP, thereby substantially reducing data leakage risks.  
\end{enumerate}

\begin{figure}[t]
\centering

\begin{subfigure}[t]{0.9\linewidth}
  \centering
  \includegraphics[width=\linewidth]{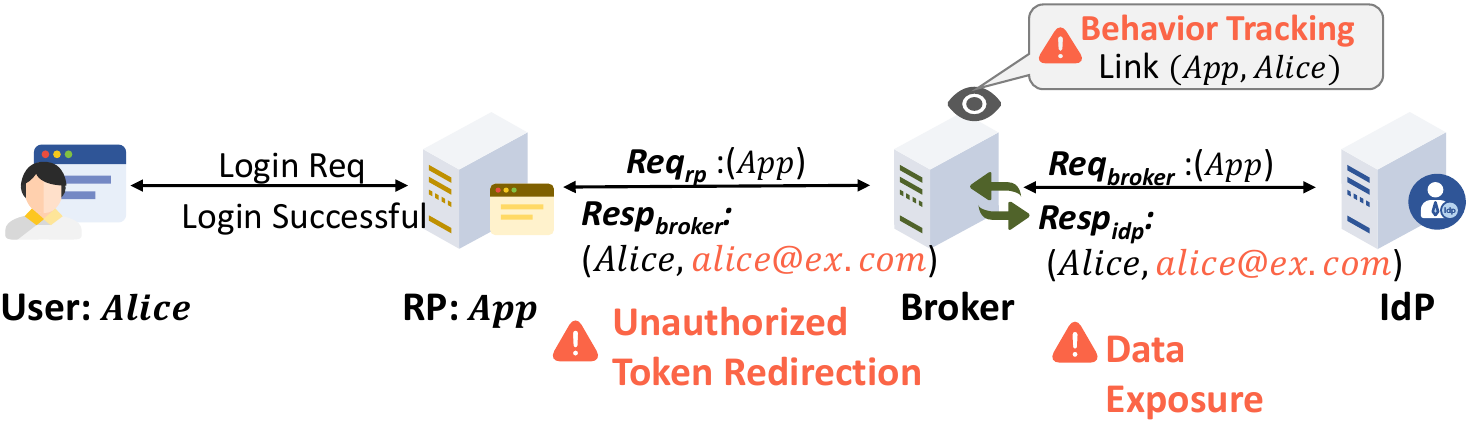}
  \subcaption{Before}
\end{subfigure}

\vspace{0.6em}

\begin{subfigure}[t]{0.9\linewidth}
  \centering
  \includegraphics[width=\linewidth]{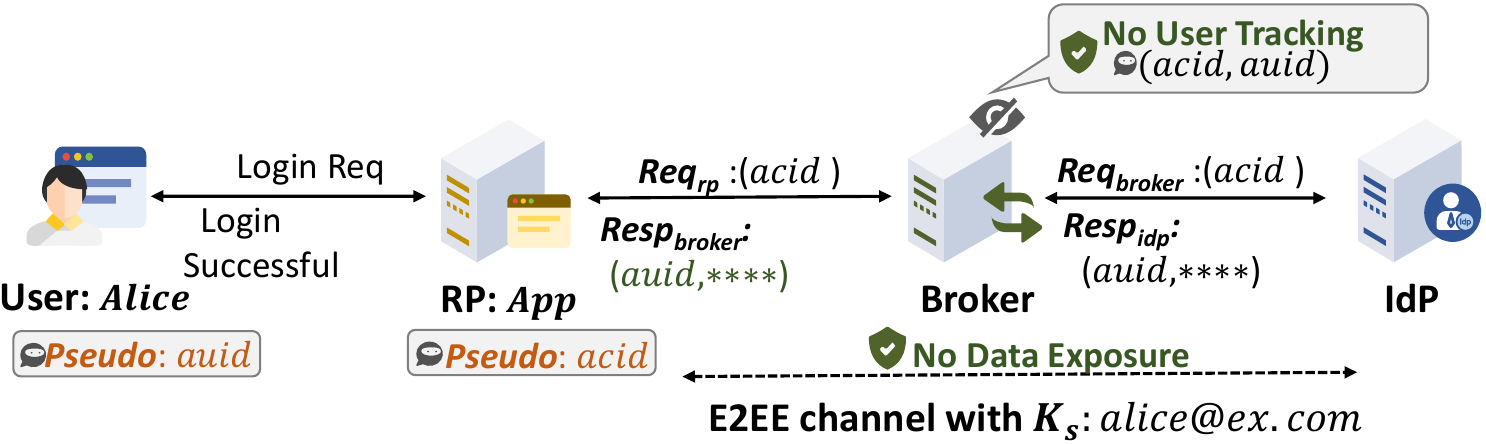}
  \subcaption{\name{}}
\end{subfigure}

\vspace{-0.3em}
\caption{Improving Privacy: Before vs.\ \name{}.}
\label{fig:improving-privacy}
\end{figure}

\begin{table}[t]
\centering
\caption{Broker Information Exposure: Before vs.\ \name{}.}
\label{tab:broker-privacy}

\vspace{-0.3em}

\begin{threeparttable}
\begin{footnotesize}

\newcommand{\riskicon}{\faExclamationTriangle}
\newcommand{\shieldicon}{\faIcon{shield-alt}}
\newcommand{\partialicon}{\faAdjust}

\begin{tabular}{lcc}
\toprule
\textbf{Information Type} & \textbf{Before} & \textbf{\name{}} \\
\midrule
Unique user identifier & \riskicon & \shieldicon\tnote{1} \\
Cross-RP linkability & \riskicon & \shieldicon \\
User attributes (e.g., email, name, phone) & \riskicon & \shieldicon \\
Access token & \riskicon & \shieldicon \\
Refresh token & \riskicon & \partialicon\tnote{2} \\
\bottomrule
\end{tabular}

\begin{tablenotes}
\footnotesize
\item[1] The broker observes only an RP-specific pseudonym rather than the user’s real identifier.
\item[2] The broker receives a refresh token but cannot access the refreshed access token or its contents.
\end{tablenotes}

\end{footnotesize}
\end{threeparttable}
\end{table}

\re{Throughout our design, we assume an \emph{honest-but-curious} identity broker: 
the broker correctly follows the protocol specification but attempts to infer sensitive information from all data it can observe. 
This threat model captures a broad range of realistic risks, including data breaches, as any information visible to the broker is assumed to be ultimately accessible to an adversary. 
Our model therefore reflects the practical risk that broker-observed information may be exposed through channels outside the primary authentication logic.
Designing brokered SSO protocols secure against a fully malicious broker is an important direction for future work and lies beyond the scope of this paper.}

In addition, we instantiate \name{} on top of the widely deployed OIDC authorization code flow \cite{sakimura2014openid} and introduce two new features to improve practical usability:
(\emph{a}) \emph{private authorization management} and
(\emph{b}) \emph{forward-secure token refresh}.

The first feature enables users to review and revoke RP accesses without revealing authorization relationships to the IdP. 
This addresses a critical yet often overlooked gap in prior privacy-preserving SSO systems \cite{he2024arpsso,heher2025bison,xu2023miso,kroschewski2025oppid}: 
although anonymous IdP-to-RP authentication prevents the IdP from tracking users’ authorization choices, 
users still require visibility and control over their authorized RPs for tasks such as access review and revocation.  
The second feature allows the broker to assist RPs in securely maintaining token validity using refresh tokens.
Each refresh response is encrypted under a freshly derived session key, providing forward secrecy and post-compromise security.

\name{} is usable and efficient in practice, and it remains largely transparent to the user.
Our evaluation shows that \name{} incurs an end-to-end login latency of 2.08 seconds in a WAN setting, with additional local computation overhead below 0.07 seconds.
Compared to a standard brokered SSO pipeline, which typically requires approximately 0.5 seconds per login, \name{} introduces only modest overhead while providing substantially stronger security and privacy guarantees.
These results indicate that \name{} meets real-world performance requirements and is suitable for practical deployment.

In summary, we make the following contributions: 

\begin{itemize}[leftmargin=*]
\item \textbf{Problem formulation.} 
We identify \emph{data exposure at identity brokers} as a fundamental security and privacy risk in brokered SSO systems 
and formalize the brokered SSO pipeline through a rigorous system model. 

\item \textbf{System design.} 
We design \name{}, a new brokered SSO pipeline that enforces \emph{broker obliviousness} through cryptographic constructions, ensuring that brokers neither learn user identities nor access sensitive authentication data, thereby preventing data leakage in a principled manner.
We further instantiate \name{} on OIDC and introduce two usability-enhancing features. 

\item \textbf{Formal security analysis.} 
We formally define the security and privacy goals of broker obliviousness and provide rigorous cryptographic proofs.

\item \textbf{Implementation and evaluation.} 
We implement and open-source a prototype of \name{} and evaluate its performance and usability, demonstrating that it achieves strong security guarantees with only modest overhead compared to standard brokered SSO pipelines, making it suitable for real-world deployment. 
\end{itemize}

\section{Preliminaries and Building Blocks}
In this section, we outline the preliminaries and foundational blocks necessary for \name{}.

\noindent\textbf{Notations:}  
Let $\mathbb{G}$ be a cyclic group of prime order $p$ with generator $g$, and denote $\mathbb{G}^* = \mathbb{G} \setminus \{1_{\mathbb{G}}\}$. 
All group operations are written multiplicatively, with scalar exponentiation denoted as $X^a$ for $X \in \mathbb{G}$ and $a \in \mathbb{Z}_p$.

\subsection{Brokered SSO Flow}\label{sec:Brokered SSO Flow}

Brokered SSO enables RPs to connect to IdPs through a single broker, which provides a range of functionalities to reduce deployment and maintenance overhead, including:
\begin{enumerate}[leftmargin=*]
    \item \textit{Unified Authentication Interface.}
Exposes a single standardized authentication interface, allowing applications to support heterogeneous IdPs without IdP-specific logic.
\item \textit{Identity Normalization.}
Maps identities from diverse IdPs into a unified user representation, enabling RPs to operate on stable and standardized claims.
\item \textit{Centralized Security Enforcement.}
Centralizes security-critical mechanisms, including token validation,
token refresh, MFA, and risk control.
\item \textit{Scalability.}
Allows IdPs to be added or updated at the broker without requiring
RP-side changes.
\end{enumerate}

Figure~\ref{fig:brokerSSOflow} shows the standard brokered SSO pipeline, consisting of a \textit{Registration} phase and a \textit{Login Flow} phase.
Overall, the broker performs two standard SSO flows independently: first acting as a proxy RP to authenticate with the IdP, and then as a proxy IdP to authorize the RP.

\noindent\textbf{Phase 1: Registration.} 
\circnum{1} The RP registers with the IdP and obtains \texttt{client\_idb} (i.e., \texttt{client\_id} and \texttt{client\_secret}).  
\circnum{2} The RP forwards \texttt{client\_idb} to the broker to register and receives \texttt{client\_rp}.  

\noindent\textbf{Phase 2: Login Flow.} 
\circnum{1}–\circnum{2} When the user selects \textit{``Login with IdP''} on the RP, the RP redirects the authentication request to the broker.  
\circnum{3} The broker generates a standardized request according to the specific IdP and forwards it to the IdP to initiate authentication.  
\circnum{4} The user verifies the RP’s identity on the frontend and decides whether to grant authorization.
\circnum{5} After authenticating the user and obtaining consent, the IdP returns a token containing the user’s personal information (\texttt{userinfo}) to the broker.  
In the OIDC Authentication Code flow, the broker additionally exchanges the authorization code for an \texttt{id\_token} or \texttt{access\_token}.  
\circnum{6} The broker verifies the token signature, parses its content, and records the RP–user mapping. Then the broker issues a new token for the RP and redirects the browser back to the RP.  
\circnum{7} The RP validates the token, establishes a local session, and grants the user access.  

\begin{figure}[t]
    \centering
    \includegraphics[width=\linewidth]{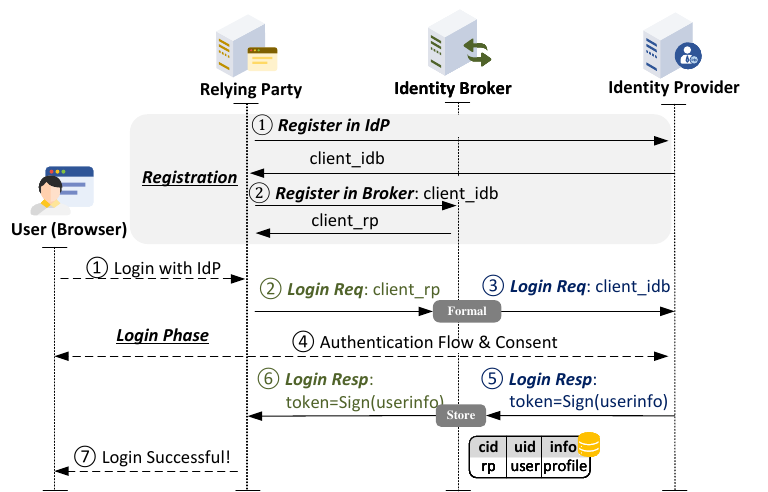}    
    \caption{Brokered SSO Login Flow.}
    \label{fig:brokerSSOflow}
\end{figure}

\subsection{Oblivious Pseudo-Random
Functions} \label{sec:OPRF}

An OPRF~\cite{casacuberta2022sok,davidson2023rfc} is a two-party protocol where the server holds a PRF key $k$, and the client provides an input $X$. 
The client learns only the output $\mathsf{OPRF}_k(X)$, while the server learns nothing about $X$ or the final output.

We adopt OPRFs over prime-order groups~\cite{davidson2019oblivious,jarecki2014round}:

\begin{itemize}[leftmargin=*]
    \item $(r,\alpha) \leftarrow \mathsf{Blind}(x)$: 
    The client samples a random blinding factor $r \xleftarrow{\$} \mathbb{Z}_p^*$, computes the blinded element $\alpha = H(x)^r$, and outputs $(r,\alpha)$, where $H: \{0,1\}^* \rightarrow \mathbb{G}$ is a hash-to-group function. 
    
    \item $\beta \leftarrow \mathsf{BlindEval}(k,\alpha)$: 
    The server computes $\beta = \alpha^k$ using its secret key $k \in \mathbb{Z}_p^*$ and returns $\beta$. 
    
    \item $y \leftarrow \mathsf{Unblind}(r,\beta)$: 
    The client computes $y=\beta^{1/r}=H(x)^k$ to remove the blinding from the input value $\beta$, and outputs unblinded element $y$. In practice, $y$ is often hashed together with the input $x$ to produce the final OPRF output.  
\end{itemize}
This protocol satisfies \textit{correctness}, \textit{consistency}, and 
\textit{pseudorandomness} for which we provide formal definitions
in Appendix~\ref{sec:Security Properties of OPRFs}.

\subsection{Proof of Knowledge}
A Proof of Knowledge (PoK) allows a prover to convince a verifier of knowledge of a secret witness without revealing it. 
We use non-interactive PoKs derived from $\Sigma$-protocols~\cite{damgaard2002sigma} via the Fiat–Shamir transform~\cite{fiat1986prove}. 
A non-interactive zero-knowledge PoK is represented as:
$$
\mathsf{PoK} \{(x_1, x_2, ...) : \text{statements about } x_1, x_2, ... \},
$$
where $(x_1, x_2, ...)$ are secret values that, together with public parameters, satisfy the specified statement.

\subsection{Anonymous Credentials} 
\label{sec:PS Signature}  

Anonymous credentials~\cite{kakvi2023sok} allow a \textit{prover} to obtain credentials on attributes from an \textit{issuer} and later prove possession or selectively disclose attributes to a \textit{verifier}, preserving unlinkability. We use a Pointcheval-Sanders (PS) signature-based anonymous credential scheme~\cite{pointcheval2016short}, summarized below. 
\ret{We discuss why \name{} chooses PS signatures in Section~\ref{sec:discussion}.}
\subsubsection{Pointcheval-Sanders Signature}  
\begin{itemize}[leftmargin=*]
    \item $pp \leftarrow \mathsf{PS.Setup}(1^\lambda)$: Given a security parameter $\lambda$, select a type-3 bilinear map $e: \mathbb{G}_1 \times \mathbb{G}_2 \rightarrow \mathbb{G}_T$ with generators $g_1 \in \mathbb{G}_1$, $g_2 \in \mathbb{G}_2$, 
    where no efficient homomorphism exists between $\mathbb{G}_1$ and $\mathbb{G}_2$. 
    The public parameters are $pp = (e, \mathbb{G}_1, \mathbb{G}_2, \mathbb{G}_T, g_1, g_2)$.

    \item $(sk,vk) \leftarrow \mathsf{PS.KeyGen}(pp,n)$: 
    For $n$ attributes, randomly select $x, y_1, ..., y_n \xleftarrow{\$} \mathbb{Z}_p$, compute $X=g_1^x, Y_i=g_1^{y_i}, \bar{X}=g_2^x, \bar{Y}_i=g_2^{y_i}$ for $i=1,...,n$. 
    The keypair is secret key $sk=(x, y_1,...,y_n, X)$, and 
    verification key $vk=(Y_1,...,Y_n,\bar{X},\bar{Y}_1,...,\bar{Y}_n)$.

    \item $\sigma \leftarrow \mathsf{PS.Sign}(sk,A):$ Given attribute vector $A=(A_1, \ldots, A_n)$, select a random $h \narrowxleftarrow{\$} \mathbb{G}^*_1$ and compute signature
   $\sigma=(\sigma_1, \sigma_2)=(h, h^{x+\sum_{i=1}^n A_i y_i})$. 

    \item $(0/1) \leftarrow \mathsf{PS.Verify}(vk,\sigma,A)$: Parse $\sigma$ as $(\sigma_1, \sigma_2)$, check $e(\sigma_2, g_2)=e(\sigma_1, \bar{X} ) \prod_{i=1}^n e(\sigma_1, \bar{Y}_i)^{A_i}$.

   \end{itemize}

\noindent A key property of PS signatures is \textit{unlinkability:} 
a PS signature $\sigma=(\sigma_1, \sigma_2)$ can be randomized by sampling $k \xleftarrow{\$} \mathbb{Z}_p$ and computing $\sigma' = (\sigma_1^k, \sigma_2^k)$, 
which remains a valid signature on the same attributes. 

\subsubsection{Anonymous Credential Protocol}  The protocol adopts the same setup and key
generation procedures as the PS signature.
\begin{itemize}[leftmargin=*]
    \item $pp \leftarrow \mathsf{AC.Setup}(1^\lambda)$:
    Run $\mathsf{PS.Setup}(1^\lambda)$ to generate the public parameters $pp$.
    \item $(sk,vk) \leftarrow \mathsf{AC.KeyGen}(pp,n)$:
    Run $\mathsf{PS.KeyGen}(pp,n)$ to generate the issuer's secret key $sk$
    and verification key $vk$.
\end{itemize}
The protocol proceeds in three phases: 
Issuance, Presentation, and Verification.

\textbf{Issuance.} $(A,\sigma) \leftarrow  \mathsf{AC.IssueCred}(pp,sk,A) $. 
The prover obtains a credential from the issuer for attribute vector $A=(A_h,A_p)$, hiding private attributes $A_h=(A_1,...,A_j)$ while disclosing public attributes $A_p=(A_{j+1},...,A_n)$: 
\begin{itemize}[leftmargin=*]
    \item $(d,\pi_1) \leftarrow \mathsf{PrepareBlindSign}(vk,A_h)$: 
    Prover selects $d \xleftarrow{\$} \mathbb{Z}_p^*$, generates a non-interactive zero-knowledge (NIZK) proof $\pi_1 = \mathsf{PoK}\{(d,A_h): C=g_1^d \prod_{i=1}^j Y_i^{A_i} \wedge \phi_1(A_h)=1\}$ that hidden attributes ($A_h$) satisfy relation $\phi_1$. 
  
    \item $\sigma' \leftarrow \mathsf{BlindSign}(sk,A_p,\pi_1)$: 
    Issuer verifies $\pi_1$, extracts $C$, selects $r \xleftarrow{\$} \mathbb{Z}_p$, and computes the signature $\sigma' = (g_1^r, (X C \prod_{i=j+1}^n Y_i^{A_i})^r)$, then sends it to the prover.   
   
    \item $\sigma \leftarrow \mathsf{UnblindSign}(d,\sigma')$: Prover computes $\sigma=(\sigma'_1, \sigma'_2  \sigma_1^{'-d})$ to obtain the final credential $\sigma$.
\end{itemize}

\noindent\textbf{Presentation.} $(A'_p,\sigma' ,\pi_2,\phi) \leftarrow \mathsf{AC.PresentCred}(vk,\sigma, A)$.
To prove possession of a credential, the prover re-randomizes $\sigma$ and selectively discloses public attributes ($A'_p$), generating NIZK proof $\pi_2$ that hidden attributes ($A'_h$) satisfy relation $\phi$.
Specifically, it samples $k,t \xleftarrow{\$} \mathbb{Z}_p^*$ and computes
$$\sigma' = (\sigma_1',\sigma_2') = (\sigma_1^k, (\sigma_2  \sigma_1^t)^k).$$

The prover then generates a NIZK proof
$$\pi_2 = \mathsf{PoK}\{(t,A_h'): \mathsf{PS.Verify}(vk, (\sigma_1', \sigma_2' \sigma_1^{'-t}), A) \wedge \phi_2(A_h')=1\},$$%
which demonstrates knowledge of a valid signature on the entire attribute vector $A$, while keeping hidden attributes $A_h'$ secret.

\noindent\textbf{Verification.} $(0/1) \leftarrow \mathsf{AC.VerifyCred}(vk,A'_p,\sigma',\pi_2,\phi)$.
The verifier checks the NIZK proof $\pi_2$ to ensure that the credential is valid under the public attributes $A'_p$ and that the hidden attributes $A'_h$ satisfy the relation $\phi_2$.

\re{When no attribute is hidden ($A'_h=\emptyset$), $\pi_2$ degenerates to revealing $t$ itself, and the verifier directly checks
$\mathsf{PS.Verify}(vk,(\sigma_1',\sigma_2' \cdot \sigma_1'^{-t}),A)=1$.
This protocol satisfies \textit{correctness}, \textit{unlinkability},
\textit{blindness}, and \textit{unforgeability}, for which we provide formal definitions
in Appendix~\ref{sec:Security Properties of Anonymous Credential}.}
\section{Threat Model and Target Properties}
This section first introduces the threat model and then outlines the security and privacy properties, as well as the security assumptions, targeted by our proposed framework.

\subsection{Threat Model}
\label{sec:Threat Model}

We consider a setting with malicious users and malicious RPs, as well as honest-but-curious identity brokers and IdPs.
Corruption is static, and the IdP colludes neither with the broker nor with any RP. The broker may share its view with corrupted RPs while still following the protocol.

\begin{itemize}[leftmargin=*]
    \item \textbf{Malicious Users.} 
    An adversary may control a user and attempt to impersonate another honest user to authenticate to honest RPs. 

    \item \textbf{Malicious RPs.} 
    An adversary may control one or more RPs and attempt to impersonate honest RPs to obtain user authorization or retrieve user data from the broker without proper consent. 
    Colluding malicious RPs may further attempt to link a user’s identity or activities across multiple services.

    \item \textbf{Honest-but-Curious IdP.} 
    The IdP is assumed to follow the protocol correctly but may attempt to infer additional information, such as tracking a user’s behavior across different RPs.
   
    \item \textbf{Honest-but-Curious Broker.} 
    The broker is assumed to behave honestly due to reputational and operational incentives but may passively collect and analyze all authentication and authorization data it observes, including information provided by IdPs. 
    For example, the broker may attempt to reconstruct users’ behavioral patterns across multiple RPs and IdPs. 
    In addition, mismanagement or misconfiguration may inadvertently expose user data to unauthorized RPs or third parties (see Appendix~\ref{app:Measurements of Vulnerabilities} for details). 

\end{itemize}

\subsection{Target Properties}
\label{sec:Target Properties} 

OrionLink aims to preserve the standard security and privacy guarantees of brokered SSO while achieving the stronger goal of \textit{broker obliviousness}. 
The ideal functionality $\mathcal{F}_{\mathsf{OrionLink}}$, shown in Figure~\ref{fig:ideal-orionlink}, captures the desired security and privacy guarantees by structuring SSO interactions into a sequence of well-defined stages and explicitly specifying the information revealed at each stage.  

\noindent \textbf{Standard SSO Security.}
\textit{OrionLink} satisfies the following standard SSO security properties: 

\begin{itemize}[leftmargin=*]
\item \textbf{Authentication Correctness.}
An adversary cannot impersonate an honest user or an honest RP to successfully authenticate through an honest IdP. 

\item \textbf{Authorization Correctness.}
An adversary cannot obtain or use protected resources without a valid authorization explicitly granted by the user to the RP. 

\item \textbf{Session Integrity.}
An honest RP accepts a login result or associated authorization credentials only if the user has intentionally authenticated that RP within the same protocol execution.

\end{itemize}

\noindent \textbf{Broker Obliviousness.} 
The broker learns no information beyond the explicitly permitted metadata and the length of transferred user information, as specified by the $\mathsf{BrokerView}$ outputs in Figure~\ref{fig:ideal-orionlink}.
We formalize this property in Definition~\ref{def:broker-obliviousness}.

\noindent \textbf{Unlinkability and Untraceability.} 
The RP remains anonymous to the IdP, preventing the IdP from learning which RP a user interacts with, as formalized in Definition~\ref{def:untraceability}.
In addition, user identifiers are unlinkable across RP–broker pairs, protecting users from cross-domain tracking, as formalized in Definition~\ref{def:unlinkability}.

\noindent \textbf{Compatibility \& Usability.} 
OrionLink remains deployable within existing brokered SSO ecosystems without requiring client-side modifications, while preserving the broker’s assistive functionalities, as discussed in Section~\ref{sec:discussion}.

\begin{figure*}[t]
\begin{mdframed}[innerleftmargin=3.5pt, innerrightmargin=3.5pt,
                 innertopmargin=3pt, innerbottommargin=3pt]
\begin{center}
    \textbf{$\mathcal{F}_{\mathsf{OrionLink}}$}
\end{center}

\textbf{Notation.}
Let $D$ denote the RP identity (e.g., domain),
$cid$ an RP identifier,
$acid$ a session-specific RP pseudonym,
$uid$ a global user identifier,
$uid_{rp}$ an RP-specific user pseudonym,
$info$ authentication credentials and user attributes,
and $sid$ a session identifier.

\begin{itemize}[leftmargin=*]

\re{\item \textbf{On receiving} $(\mathsf{RegisterRP}, RP, D)$:
\begin{enumerate}[leftmargin=5mm]
  \item Sample a fresh random identifier $cid$. Set $\mathsf{RPReg}[cid] := D$ and record the registrant $\mathsf{Owner}[cid] := RP$.
  \item Output $(\mathsf{Registered}, cid)$ to $RP$,
        $(\mathsf{Registered}, cid, D)$ to $IdP$,
        and $(\mathsf{BrokerView}, cid, D)$ to $Broker$.
\end{enumerate}

\item \textbf{On receiving} $(\mathsf{InitLogin}, RP, sid, cid, acid)$:
\begin{enumerate}[leftmargin=5mm]
  \item If $\mathsf{RPReg}[cid] = \bot$ or $\mathsf{Sess}[sid] \neq \bot$, abort.
  \item If $RP$ is honest, replace the supplied $acid$ with a fresh random one.
  \item If $RP$ is corrupted, abort unless $\mathsf{Owner}[cid]$ is corrupted. Then send $(\mathsf{Judge}, sid, cid, acid)$ to the adversary and abort unless it replies $\mathsf{accept}$, in which case adopt the supplied $acid$.
  \item Set $\mathsf{Sess}[sid] := (\mathsf{pending}, cid, acid)$.
  \item Output $(\mathsf{LoginCtx}, sid, acid)$ to $RP$,
        $(\mathsf{LoginReq}, sid, acid)$ to $IdP$,
        and $(\mathsf{BrokerView}, sid, cid, acid)$ to $Broker$.
\end{enumerate}}

\item \textbf{On receiving} $(\mathsf{ShowRP}, RP, sid, cid, D)$:
\begin{enumerate}[leftmargin=5mm]
  \item If $\mathsf{Sess}[sid] \neq (\mathsf{pending}, cid, \cdot)$ or $\mathsf{RPReg}[cid] \neq D$, abort.
  \item Output $(\mathsf{ShowRP}, sid, D)$ to $User$.
\end{enumerate}

\item \textbf{On receiving} $(\mathsf{Authorize}, User, sid, uid, decision)$:
\begin{enumerate}[leftmargin=5mm]
  \item If $\mathsf{Sess}[sid] \neq (\mathsf{pending}, cid, acid)$ for some $(cid, acid)$, abort.
  \item If $decision = \mathsf{approve}$, update $\mathsf{Sess}[sid] := (\mathsf{approved}, cid, acid)$, otherwise abort.
  \item Output $(\mathsf{Consent}, sid, uid)$ to $IdP$.
\end{enumerate}

\item \textbf{On receiving} $(\mathsf{FinalizeLogin}, IdP, sid, uid, acid, info)$:
\begin{enumerate}[leftmargin=5mm]
  \item If $\mathsf{Sess}[sid] \neq (\mathsf{approved}, cid, acid)$ with the same $acid$ as supplied, abort. Otherwise parse the recorded $cid$.
  \item If $\mathsf{UID}[uid,cid] = \bot$,
        sample fresh $uid_{rp}$ and set $\mathsf{UID}[uid,cid] := uid_{rp}$.
  \item Output $(\mathsf{AuthOK}, sid, uid_{rp}, info)$ to $RP$,
        and $(\mathsf{BrokerView}, sid, uid_{rp}, |info|)$ to $Broker$.
  \item Set $\mathsf{Sess}[sid] := \mathsf{done}$.
\end{enumerate}

\end{itemize}
\end{mdframed}
\caption{Ideal functionality $\mathcal{F}_{\mathsf{OrionLink}}$.}
\label{fig:ideal-orionlink}
\end{figure*}

\subsection{Assumptions}
\label{sec:Assumptions}

Our design is based on the following security and operational assumptions:

\begin{itemize}[leftmargin=*]
    \item \textbf{Broker Visibility.} 
    It is acceptable for an IdP to learn an RP's broker, as brokers typically serve multiple RPs (see Appendix~\ref{app:Measurements of Ecosystem}). 
    This does not compromise RP anonymity. 
    Accordingly, we do not require additional network-layer anonymity mechanisms (e.g., Tor~\cite{dingledine2004tor}) to mitigate IP-based linkability, 
    improving usability and deployability.

    \item \textbf{Unique User Identifier.}  
    We assume that the unique identifier ($uid$) is the sole uniquely identifying user attribute maintained by the IdP. 

\end{itemize}

\section{OrionLink}

OrionLink is a privacy-enhancing brokered SSO system, 
involving four entities: the user, the RP, the IdP, and the broker. 
We describe \name{} based on the standard Brokered SSO workflow of Section \ref{sec:Brokered SSO Flow}, which is compatible with common SSO protocols, including OIDC \cite{sakimura2014openid} and SAML 2.0 \cite{hughes2005security}.
\subsection{Overview}

The simplest route to broker obliviousness is end-to-end encryption between the RP and the IdP, with the broker reduced to a forwarder. This \emph{strawman} hides user data from the broker, but it does not by itself hide the identities used to establish the channel. In a conventional end-to-end encrypted design, the RP and the IdP still authenticate using long-term or otherwise linkable credentials. As a result, the user-RP authorization relationship remains linkable and traceable, which breaks our security targets in \S\ref{sec:Target Properties}. A second design axis, \emph{anonymous endpoints}, is therefore unavoidable.

\name{} couples anonymous authentication with the encrypted channel through two mechanisms.
(\emph{a})~\textit{Anonymous pseudonym derivation} replaces long-term RP identifiers in IdP-facing interactions with fresh session-specific pseudonyms, while mapping global user identifiers to stable, RP-specific pseudonyms.
(\emph{b})~\textit{Anonymous authentication key agreement (AAKA)} lets an anonymized RP and the IdP derive a session key $K_S$ over the broker without exposing any long-term key or stable party identity. We detail the construction of both mechanisms in \S\ref{sec:Design}.

\noindent\textbf{Flow Overview} (Fig.~\ref{fig:overview of OrionLink}). At a high level, \name{} proceeds in four phases as outlined below. We defer the full protocol specification to \S\ref{sec:System Details}.
\begin{itemize}[leftmargin=*]
  \item \emph{Registration} (\circnum{1}--\circnum{4}). The RP registers its identity $D$ at the IdP, which verifies the RP and issues an RP identifier $cid$ together with an anonymous credential $\mathcal{DC}$. The RP then registers $(cid,\mathcal{DC})$ at the broker. This phase mirrors the standard brokered SSO registration flow, with the only difference being that the RP additionally obtains the anonymous credential $\mathcal{DC}$ at the IdP for use in subsequent logins.
  \item \emph{Login Request} (\circnum{5}--\circnum{6}). On a user login, the RP derives a session pseudonym $acid$ from $\mathcal{DC}$, generates a fresh session key $K_S$, and sends $(cid,acid)$ to the broker, which formats and forwards $(acid)$ to the IdP. The RP still reaches multiple IdPs through the broker's unified authentication interface, preserving one of the broker's key assistive roles.
  \item \emph{Auth \& Consent} (\circnum{7}). The user independently verifies that $acid$ commits to $D$, which is displayed on the consent page, and then authenticates and grants consent at the IdP.
  \item \emph{Login Response} (\circnumtwo{8}--\circnumtwo{11}). The IdP completes the AAKA to fix $K_S$, derives the RP-specific user pseudonym $auid$, and encrypts the user info as $\mathsf{info}_{Enc}=\mathsf{Enc}_{K_S}(\mathsf{info})$. It returns $(auid,\mathsf{info}_{Enc})$ through the broker over an E2EE channel. The RP unblinds $auid$ into $uid_{rp}$ and decrypts $\mathsf{info}$ under $K_S$. The RP then returns $uid_{rp}$ to the broker, which records it as the stable per-RP user pseudonym for user management.
\end{itemize}

\noindent\textbf{Notations:}
We use $\mathsf{Enc}_K(\cdot)$ and $\mathsf{Dec}_K(\cdot)$ to denote symmetric authenticated encryption and decryption under key $K$, respectively, and $\mathsf{AEAD.Enc}_K(\cdot,\cdot)$ and $\mathsf{AEAD.Dec}_K(\cdot,\cdot)$ to denote authenticated encryption with associated data (AEAD). 
We denote $\mathsf{KDF}(\cdot)$ as a key derivation function, and $\mathsf{Sig}{sk}(\cdot)$ as a digital signature generated using the signing key $sk$.

\subsection{Design}\label{sec:Design}

This section presents the construction of \textit{anonymous pseudonym derivation} and \textit{AAKA}, which together realize the security and privacy targets of \S\ref{sec:Target Properties}.

To facilitate cryptographic operations on the client-side without requiring browser extensions, we define a \textit{Trusted Computation Agent (TCA)}. 
The TCA is a JavaScript component loaded into the user's browser via an \texttt{iframe} from a publicly accessible, auditable source. 
It assists the user in performing essential computational and verification steps, such as verifying credentials and deriving keys.
The design draws inspiration from \cite{he2024arpsso}.

\noindent\textbf{Remark.} Revealing the RP’s identity to the broker is acceptable and does not violate our security properties (see Appendix~\ref{app:additional discussion}). Hiding RP identity from the broker is left for future work under stronger assumptions.
\subsubsection{Anonymous Pseudonym Derivation}

To achieve identity obliviousness, \name{} must ensure two properties: (1) \textit{RP pseudonym derivation:} the IdP should authenticate the RP using a session-specific RP pseudonym without learning user–RP relationships, preventing cross-domain tracking. 
(2) \textit{User pseudonym derivation:} each user should be identified by a unique, RP-specific pseudonym, which prevents identity linkability by the RP and the broker, and preserves the session consistency needed for broker-side data management. While prior work \cite{he2024arpsso,heher2025bison,xu2023miso,fett2015spresso,lin2025understanding} has addressed this in standard SSO, 
we adapt these concepts for the brokered model.

\noindent\textbf{RP pseudonym derivation.} 
An RP’s identity consists of two components. The first is a human-readable identity descriptor $D$, such as a domain name, which can be directly inspected by the user to verify that the RP is indeed the intended party to be authorized.
The second is a long-term identity secret $\mathsf{rsk}_{IK}$, held exclusively by the RP, which is used to prove possession of a registered identity to the IdP, such as the  \texttt{client\_secret} in OIDC flow.

During SSO, the RP’s identity must be verified twice:

\begin{itemize}[leftmargin=*]
    \item \textit{User-facing verification}. The user verifies the RP’s identity on the frontend and decides whether to grant authorization which relies on presenting the identity descriptor $D$ to the user (step~\circnum{7}).
    \item \textit{Token Issuance}. The RP anonymously proves possession of the corresponding $\mathsf{rsk}_{IK}$ to the IdP, thereby establishing identity validity without revealing its real identity. To ensure \textit{session integrity}, the IdP must further verify that the entity requesting the authentication credentials is the same entity that the user has authorized (step~\circnumtwo{8}).
\end{itemize}
Since users do not need to verify $\mathsf{rsk}_{IK}$, we issue two separate credentials so that the user only verifies $D$ on the frontend. Building on the anonymous credential scheme from Section~\ref{sec:PS Signature}, the RP obtains two credentials from the IdP during registration: 

\begin{itemize}[leftmargin=*]
    \item \textit{Domain Credential ($\mathcal{DC}$)}: 
    A PS signature on the public attribute $D$, binding the RP to its registered identity.  
    \item \textit{Secret Credential ($\mathcal{SC}$)}: 
    A PS signature over both the public attribute $D$ and the private attribute $\mathsf{rsk}_{IK}$, which is also used in the AAKA protocol for key agreement.  

\end{itemize}

For \textit{User-facing verification}, the RP derives a pseudonymous identifier $k,t\narrowxleftarrow{\$} \mathbb{Z}^*_p, acid :=(\sigma_1^k,(\sigma_2 \sigma_1^t)^k)$ from $\mathsf{AC.PresentCred}(\mathcal{DC})$ (step~\circnum{5}) and the broker redirects it to the IdP window.
Since $D$ is fully disclosed to the user at consent time, this presentation is a fully disclosed one in the sense of Section~\ref{sec:PS Signature}, so the accompanying proof $\pi_D$ is realized by the randomization factor $t$ itself, which the RP sends alongside $acid$ and the broker forwards to the TCA.
We use TCA on behalf of the user to derandomize $acid$ with $t$, verify $D$ against it, and display $D$ on the consent page (step~\circnum{7}). Upon user confirmation, the consent result is sent to the IdP to obtain the authentication credential.
For \textit{Token Issuance}, the RP generates a NIZK proof $\pi_2$ via $\mathsf{AC.PresentCred(\mathcal{DC},\mathcal{SC})}$ (step~\circnum{5}) to prove its link to $acid$ and its possession of $\mathsf{rsk}_{IK}$.
The IdP verifies it with $\mathsf{AC.VerifyCred}$ (step~\circnumtwo{8}), ensuring RP anonymity and consistency without revealing $D$ or $\mathsf{rsk}_{IK}$.

\noindent\textbf{User pseudonym derivation.}
To derive unlinkable user pseudonyms, we adopt the OPRF-based scheme in Section~\ref{sec:OPRF},
ensuring that each pseudonym is unique per RP-Broker while remaining unlinkable across domains.  
Each user ($uid$) and each RP ($cid$) is associated with a unique identifier
assigned by the IdP, from which RP-specific user pseudonyms are derived in the form $uid_{rp} = cid^{uid}$.
  
During registration, the IdP samples $cid:=h, h \narrowxleftarrow{\$} \mathbb{G}^*_1$, and embeds $h$ into the signature $\sigma=(cid,\sigma_2)$ over the RP’s identity descriptor $D$ . 

During login phase, the RP generates a session-specific pseudonym $acid = \ret{(\sigma_1',\sigma_2')} = (cid^k, (\sigma_2 \cdot cid^t)^k)$,
which is similar to the RP blinds its identifier by computing $(k, cid^k) \leftarrow
\mathsf{Blind}(cid)$ for fresh randomness $k$.
The IdP performs a
blind evaluation and computes $auid = \mathsf{BlindEval}(uid, \ret{\sigma_1'})$ (step~\circnumtwo{8}).
Upon receiving $auid$, the RP removes the blinding factor
and computes $\mathsf{Unblind}(k, auid) = auid^{1/k} = h^{uid}$, obtaining the RP-specific user identifier $uid_{rp}$ (step~\circnumtwo{10}).

This construction allows the RP to later derive the RP-specific user identifier
$uid_{rp}$ without learning the user’s global identifier $uid$, while ensuring
that the IdP does not learn the RP identifier $cid$ during authentication.

\begin{figure}[t]
    \centering
    \includegraphics[width=1\linewidth]{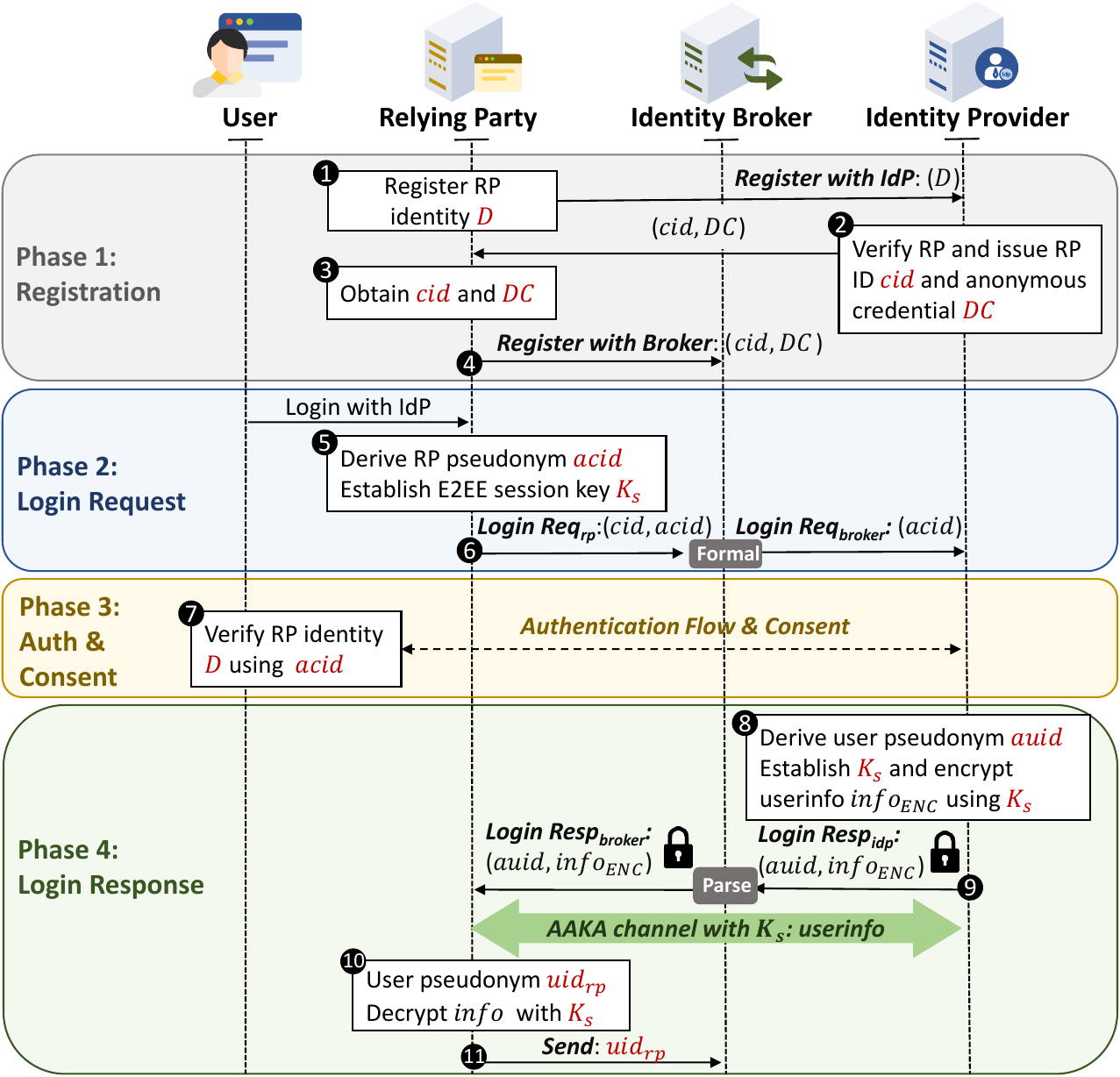}    
    \caption{Overview of OrionLink.}
    \label{fig:overview of OrionLink}
\end{figure}
\subsubsection{Anonymous Authentication Key Agreement}
After the IdP authenticates the RP and obtains user consent, it must forward the user’s authorization credentials and personal information to the RP via the broker, which may expose sensitive user data to the broker.
To prevent this, \name{} establishes an E2EE channel between the RP and the IdP using a variant of the X3DH key agreement protocol~\cite{marlinspike2016x3dh}. 
The RP and the IdP derive a fresh session key
for each session from their long-term identity keys, \ret{the RP's ephemeral key, and the IdP's signed prekey, which is refreshed once per epoch.}
\ret{Because the IdP reuses its signed prekey within an epoch, AAKA provides forward secrecy only across prekey epochs rather than at session granularity. Appendix~\ref{app:additional discussion} details this limitation and describes an upgrade to session-level forward secrecy.}
The broker only forwards ciphertexts and has no access to plaintext user data.

To resist man-in-the-middle (MitM) attacks while hiding the RP’s real identity, the RP needs to maintain a long-term identity key ($\mathsf{rsk}_{IK},\mathsf{rpk}_{IK}$) and derives a fresh, anonymous identity key ($\mathsf{rsk}'_{IK},\mathsf{rpk}'_{IK}$) for each session.
During registration, the RP embeds its long-term secret identity key $\mathsf{rsk}_{IK}$ into $\mathcal{SC}$ as a private attribute, with corresponding public key $\mathsf{rpk}_{IK}=g_1^{\mathsf{rsk}_{IK}}$.

For each session, the RP randomizes its public key as $r \narrowxleftarrow{\$} \mathbb{Z}_p, \mathsf{rpk}'_{IK}=g_1^{r + \mathsf{rsk}_{IK}}$ and proves its validity via NIZK (step~\circnum{5}).
By anonymously verifying the RP, the IdP checks that $\mathsf{rpk}'_{IK}$ is correctly derived from the RP’s long-term identity secret $\mathsf{rsk}_{IK}$,
establishing possession of a registered identity. The verified key
$\mathsf{rpk}'_{IK}$ is then used as the RP’s identity key in the key agreement
protocol and combined with \ret{the RP's ephemeral key and the IdP's signed prekey} to derive the session key
$K_S \leftarrow \mathsf{X3DH}_{\mathsf{RP}}(\mathsf{rsk}'_{IK}, \mathsf{rsk}_E, \mathsf{ipk}_{IK}, \mathsf{ipk}_{E})$
while the RP computes $K_S$ using the IdP's identity key $\mathsf{ipk}_{IK}$ obtained during the initial registration (step~\circnum{5}).
The IdP runs the mirrored responder routine $\mathsf{X3DH}_{\mathsf{IdP}}$ of Figure~\ref{fig:Anonymous Authentication Key Exchange Protocol} on the corresponding keys, so both sides derive the same $K_S$.
The IdP encrypts authorization credentials and personal information using the session key $K_S$, while keeping attribute names in plaintext. This allows the broker to map identities from diverse IdPs into a unified user representation.

\begin{figure}[t]
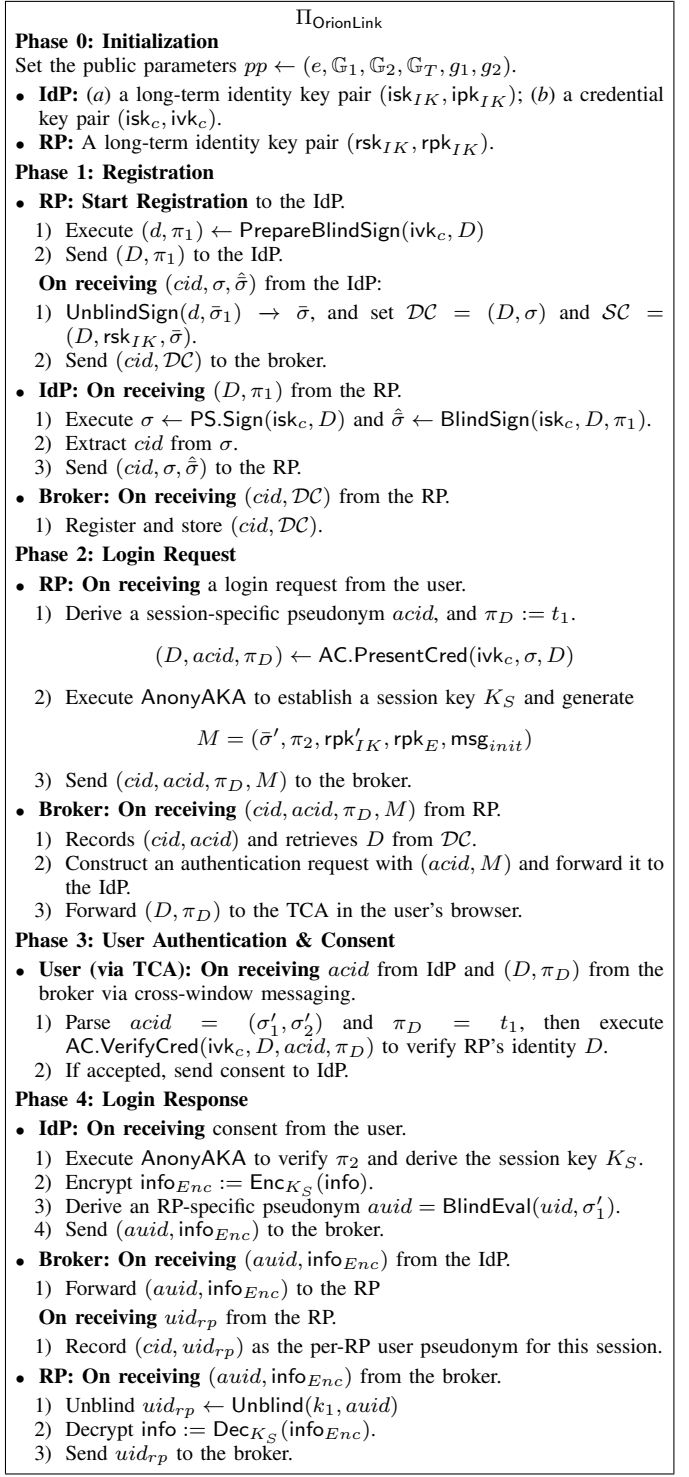

\begin{mdframed}[innerleftmargin=3.5pt, innerrightmargin=3.5pt, innertopmargin=3pt, innerbottommargin=3pt]
\begin{footnotesize}
\begin{center}
  $\Pi_{\mathsf{OrionLink}}$
\end{center}
  \textbf{Phase 0: Initialization} 
  
Set the public parameters
$pp \leftarrow (e,\mathbb{G}_1,\mathbb{G}_2,\mathbb{G}_T,g_1,g_2)$.
\begin{itemize}[leftmargin=*]
\item \textbf{IdP:}
(\emph{a}) a long-term identity key pair $(\mathsf{isk}_{IK}, \mathsf{ipk}_{IK})$;   
(\emph{b}) a credential key pair $(\mathsf{isk}_c, \mathsf{ivk}_c)$. 
\item \textbf{RP:} A long-term identity key pair $(\mathsf{rsk}_{IK}, \mathsf{rpk}_{IK})$.

\end{itemize}

  \textbf{Phase 1: Registration}  

\begin{itemize}[leftmargin=*]

\item \textbf{RP: Start Registration} to the IdP.
\begin{enumerate}
    \item Execute $(d,\pi_1) \leftarrow \mathsf{PrepareBlindSign}(\mathsf{ivk}_c,D)$
    \item Send $(D,\pi_1)$ to the IdP.
\end{enumerate}

\textbf{On receiving} $(cid,\sigma,\hat{\bar{\sigma}})$ from the IdP:

\begin{enumerate}
    \item $\mathsf{UnblindSign}(d,\bar{\sigma}_1) \xrightarrow{}\bar{\sigma}$, and set $\mathcal{DC} = (D,\sigma)$ and $\mathcal{SC} = (D,\mathsf{rsk}_{IK},\bar{\sigma})$.
    \item Send $(cid,\mathcal{DC})$ to the broker.
\end{enumerate}

\item \textbf{IdP: On receiving} $(D,\pi_1)$ from the RP.
\begin{enumerate}
    \item Execute $\sigma \leftarrow \mathsf{PS.Sign}(\mathsf{isk}_c,D)$ and $\hat{\bar{\sigma}} \leftarrow \mathsf{BlindSign}(\mathsf{isk}_c,D,\pi_1)$.
    \item Extract $cid$ from $\sigma$.
    \item Send $(cid,\sigma,\hat{\bar{\sigma}})$ to the RP.
\end{enumerate}

\item \textbf{Broker: On receiving} $(cid,\mathcal{DC})$ from the RP.
\begin{enumerate}
    \item Register and store $(cid,\mathcal{DC})$.
\end{enumerate}
\end{itemize}

\textbf{Phase 2: Login Request}

\begin{itemize}[leftmargin=*]
\item  \textbf{RP: On receiving} a login request from the user.

\begin{enumerate}
    \item \re{Derive a session-specific pseudonym $acid$, and  $\pi_D := t_1$.}
    $$(D,acid,\pi_D) \leftarrow \mathsf{AC.PresentCred}(\mathsf{ivk}_c,\sigma, D)$$
    \item Execute $\mathsf{AnonyAKA}$ to establish a session key $K_S$ and generate $$M=(\bar{\sigma}', \pi_2,\mathsf{rpk}'_{IK}, \mathsf{rpk}_E,\mathsf{msg}_{init})$$
    \item \re{Send $(cid,acid,\pi_D,M)$ to the broker.}
\end{enumerate}

\item \textbf{Broker: On receiving} $(cid,acid,\pi_D,M)$ from RP.

\begin{enumerate}
    \item Records $(cid,acid)$ and retrieves $D$ from $\mathcal{DC}$.
    \item Construct an authentication request with $(acid,M)$ and forward it to the IdP.
    \item Forward $(D,\pi_D)$ to the TCA in the user's browser.
\end{enumerate}
\end{itemize}

\textbf{Phase 3: User Authentication \& Consent}

\begin{itemize}[leftmargin=*]
\item \re{\textbf{User (via TCA): On receiving} $acid$ from IdP and $(D,\pi_D)$ from the broker via cross-window messaging.}
\begin{enumerate}
    \item \re{Parse $acid=(\sigma_1',\sigma_2')$ and $\pi_D = t_1$, then execute $\mathsf{AC.VerifyCred}(\mathsf{ivk}_c,D,acid,\pi_D)$ to verify RP's identity $D$.}
    \item If accepted, send consent to IdP.
\end{enumerate}
\end{itemize}

\textbf{Phase 4: Login Response}

\begin{itemize}[leftmargin=*]
\item \textbf{IdP: On receiving} consent from the user.
\begin{enumerate}
    \item Execute $\mathsf{AnonyAKA}$ to verify $\pi_2$ and derive the session key $K_S$.
    \item Encrypt $\mathsf{info}_{Enc} := \mathsf{Enc}_{K_S}(\mathsf{info})$.
    \item Derive an RP-specific pseudonym $auid = \mathsf{BlindEval}(uid, \ret{\sigma_1'})$.
    \item Send $(auid,\mathsf{info}_{Enc})$ to the broker.
\end{enumerate}

\item \textbf{Broker: On receiving} $(auid,\mathsf{info}_{Enc})$ from the IdP.
\begin{enumerate}
    \item Forward $(auid,\mathsf{info}_{Enc})$ to the RP
\end{enumerate}
\textbf{On receiving} $uid_{rp}$ from the RP.
\begin{enumerate}
    \item \re{Record $(cid, uid_{rp})$ as the per-RP user pseudonym for this session.}
\end{enumerate}

\item \textbf{RP: On receiving} $(auid,\mathsf{info}_{Enc})$ from the broker.
\begin{enumerate}
    \item Unblind $uid_{rp} \leftarrow \mathsf{Unblind}(k_1, auid)$
    \item Decrypt $\mathsf{info} := \mathsf{Dec}_{K_S}(\mathsf{info}_{Enc})$.
    \item \re{Send $uid_{rp}$ to the broker.}
\end{enumerate}

\end{itemize}

\end{footnotesize}
\end{mdframed}
\caption{The \textit{OrionLink} Protocol $\Pi_{\mathsf{OrionLink}}$.}
  \label{fig:The OrionLink Protocol}
\end{figure}

\subsection{Privacy without Sacrificing Broker
  Utility}
With both identifiers and user data hidden from the broker, one immediate question follows:

\textit{Does broker obliviousness come at the cost of broker utility?}

The answer is largely no. \name{} restricts only the broker functions that require plaintext values and preserves those that depend on protocol structure and metadata. The message pattern between the RP and the IdP is unchanged, so the broker still provides a unified authentication interface, protocol translation, and a single upgrade point for both sides (\textit{Unified Authentication Interface}~\cmark{}, \textit{Scalability}~\cmark{}). Attribute names stay in plaintext and the per-RP pseudonym $uid_{rp}=h^{uid}$ remains stable, so the broker can still map heterogeneous IdP claims into a unified representation and maintain user-RP relationships (\textit{Identity Normalization}~\hmark{}). Plaintext refresh metadata (Section~\ref{sec:Token Refresh}) still supports token validation, refresh, and audit (\textit{Centralized Security Enforcement}~\hmark{}).

\textbf{Limitations.}
\re{Broker obliviousness carries a functionality cost. Without plaintext values or token payloads, the broker cannot offer value-dependent services such as attribute-value normalization, broker-side attribute-based access control (ABAC), user profiling, or content-based threat analytics. We accept this cost because \name{} targets cross-organizational public brokers, which carry 71\% of observed brokered SSO flows (Appendix~\ref{app:Measurements of Ecosystem}), where user privacy takes priority. Enterprise deployments that value centralized logging or ABAC can selectively disclose attribute values to the broker, explicitly relaxing broker obliviousness for those attributes only. \name{} also adds modest key-management overhead for RPs.}

\subsection{System Details}\label{sec:System Details}
The \name{} protocol consists of four phases: \textit{Registration}, \textit{Login Request}, \textit{User Authentication \& Consent}, and \textit{Login Response}. Figure \ref{fig:overview of OrionLink} shows the overview of \name{}. Figure \ref{fig:The OrionLink Protocol} specifies the full \name{} protocol. We omit certain engineering details that do not affect the system.

\newcommand{\stepbox}[1]{%
  \tikz[baseline=(char.base)]{
    \filldraw[fill=black, draw=black] (0,0) rectangle (0.3,0.3);
    \node[text=white, font=\sffamily\small\bfseries] (char) at (0.15,0.15) {#1};
  }%
}

\subsubsection{Initialization}
\label{sec:Initialization}

In this stage, the IdP, the RP, and the broker generate the required key materials and public parameters.
 
\noindent\textbf{Public parameters:} 
Let $\lambda$ be the security parameter. 
Select a type-3 bilinear map $e: \mathbb{G}_1 \times \mathbb{G}_2 \rightarrow \mathbb{G}_T$ with generators $g_1 \in \mathbb{G}_1$ and $g_2 \in \mathbb{G}_2$. 
The public parameters are denoted as
$pp \leftarrow (e, \mathbb{G}_1, \mathbb{G}_2, \mathbb{G}_T, g_1, g_2)$.

\begin{itemize}[leftmargin=*]
\item The IdP generates the following key materials :

(\emph{a}) a long-term identity key pair $(\mathsf{isk}_{IK}, \mathsf{ipk}_{IK})$ for the AAKA protocol with the RP, where $\mathsf{isk}_{IK} \narrowxleftarrow{\$} \mathbb{Z}_p$ and $\mathsf{ipk}_{IK} = g_1^{\mathsf{isk}_{IK}}$;  

(\emph{b}) a credential key pair $(\mathsf{isk}_c, \mathsf{ivk}_c)$: the IdP chooses $x_1,y_1 \narrowxleftarrow{\$} \mathbb{Z}_p$ for $\mathcal{DC}$ and $x_2,y_{21},y_{22} \narrowxleftarrow{\$} \mathbb{Z}_p$ for $\mathcal{SC}$, then computes 
\(X_1 = g_1^{x_1}, Y_1 = g_1^{y_1}, \bar{X}_1 = g_2^{x_1}, \bar{Y}_1 = g_2^{y_1}\) and \(X_2 = g_1^{x_2}, Y_{21} = g_1^{y_{21}}, Y_{22} = g_1^{y_{22}}, \bar{X}_2 = g_2^{x_2}, \bar{Y}_{21} = g_2^{y_{21}}, \bar{Y}_{22} = g_2^{y_{22}}\). 
The IdP sets $\mathsf{isk}_c \leftarrow (x_1,y_1,X_1,\\x_2,y_{21},y_{22},X_2)$ and $\mathsf{ivk}_c \leftarrow (Y_1,\bar{X}_1,\bar{Y}_1,Y_{21},Y_{22},\bar{X}_2,\bar{Y}_{21},\bar{Y}_{22})$.

\item RP key materials: 
A long-term identity key pair $(\mathsf{rsk}_{IK}, \mathsf{rpk}_{IK})$, where $\mathsf{rsk}_{IK} \narrowxleftarrow{\$} \mathbb{Z}_p$ and $\mathsf{rpk}_{IK} = g_1^{\mathsf{rsk}_{IK}}$, for the AAKA protocol.

\end{itemize}

\subsubsection{Registration}
\label{sec:Registration}

In this stage, the user and the RP register with the IdP. 
The IdP verifies the user’s identity and assigns a unique identifier $uid \narrowxleftarrow{\$} \mathbb{Z}_p$.

The RP registers with the IdP by submitting its identity key $\mathsf{rsk}_{IK}$ as a private attribute and its identity attribute $D$ (e.g., domain) as a public attribute. 
The procedure is as follows:
\begin{itemize}[leftmargin=*]
    \item The RP selects a blinding factor $d \narrowxleftarrow{\$} \mathbb{Z}_p$ and computes a Pedersen commitment 
    $C := g_1^d \cdot Y_{21}^{\mathsf{rsk}_{IK}}$.  
    It then generates a NIZK proof
    $$\pi_1 = \mathsf{PoK}\{(d, \mathsf{rsk}_{IK}): C=g_1^d \cdot Y_{21}^{\mathsf{rsk}_{IK}}\},$$ 
    and sends $(D,\pi_1)$ to the IdP.
    \item The IdP verifies the proof $\pi_1$. If valid, it sets $cid=h$ as the RP's unique identifier, where $h \narrowxleftarrow{\$} \mathbb{G}^*_1$, and selects $r \narrowxleftarrow{\$} \mathbb{Z}_p$. It then computes a signature on $D$ as $\sigma=(\sigma_1,\sigma_2)=(h,h^{x_1+y_1 \cdot D})$, 
    and a blinded signature on $D$ and $\mathsf{rsk}_{IK}$ as 
    $\hat{\bar{\sigma}}=(g_1^{r},(X_2 \cdot C \cdot Y_{22}^D)^{r})$.
    It returns $(cid, \sigma, \hat{\bar{\sigma}})$ to the RP.
    \item The RP unblinds $\hat{\bar{\sigma}}$ to obtain 
    $\bar{\sigma}=(\hat{\bar{\sigma}}_1, \hat{\bar{\sigma}}_2  \hat{\bar{\sigma}}_1^{-d})$,  
    and derives the domain credential $\mathcal{DC}=(D,\sigma)$ and the secret credential $\mathcal{SC}=(D,\mathsf{rsk}_{IK},\bar{\sigma})$.
\end{itemize}
Finally, the RP forwards $(cid, \mathcal{DC})$ to the broker. We omit the broker’s response for brevity.

\subsubsection{Login Request}
\label{sec:Authorization Request}
In this phase, the RP receives the login request from the user, generates an RP pseudonym and the AAKA materials, and forwards them to the broker. The broker then constructs a standard SSO authorization request containing these materials and sends it to the IdP for authorization.
\vspace{0.1cm}

\noindent \textbf{Pseudonym Derivation.} The RP randomizes $\mathcal{DC}$ to derive the RP’s session-specific pseudonym.
It samples $k_1,t_1 \narrowxleftarrow{\$} \mathbb{Z}^*_p$ and computes
$acid := \sigma'=(\sigma_1^{'},\sigma_2^{'})=(\sigma_1^{k_1},(\sigma_2 \sigma_1^{t_1})^{k_1}).$
The RP also outputs $\pi_D := t_1$, the fully disclosed presentation proof that lets the user derandomize $acid$ and check it against $D$.
The blinding factor $k_1$ is never released, so $acid$ remains unlinkable across sessions.

\noindent \textbf{AAKA execution.} The RP then runs the $\mathsf{AnonyAKA}$ protocol in Figure \ref{fig:Anonymous Authentication Key Exchange Protocol} to generate $M=(\pi_2, \bar{\sigma}', \mathsf{rpk}'_{IK}, \mathsf{rpk}_E,\mathsf{msg}_{init})$ ,which includes the anonymous authentication proof $(\bar{\sigma}', \pi_2,\mathsf{rpk}'_{IK})$ and key agreement payload $(\mathsf{msg}_{init}, \mathsf{rpk}_E)$.

The RP executes 
\[
(\pi_2, \bar{\sigma}', (\mathsf{rsk}'_{IK}, \mathsf{rpk}'_{IK}))
\leftarrow
\mathsf{GenProof}(\mathcal{SC}, acid, \mathsf{ivk}_c) 
\] 
and obtains the proof: (i) a randomized credential signature $\bar{\sigma}'$ derived from $\mathcal{SC}$,  
(ii) an ephemeral identity key $\mathsf{rpk}'_{IK}$, and  
(iii) a NIZK proof $\pi_2$.
The proof guarantees that $acid$ and $\mathcal{SC}$ correspond to the same identity attribute $D$, and that $\mathsf{rpk}'_{IK}$ is correctly derived from $\mathsf{rsk}_{IK}$ embedded in $\mathcal{SC}$. 

The RP then derives the session key $K_S$ using its freshly generated ephemeral
key $\mathsf{rsk}_E$ and the \ret{signed prekey} $\mathsf{ipk}_E$ obtained from
the IdP
$$K_S \leftarrow \mathsf{X3DH}_{\mathsf{RP}}(\mathsf{rsk}'_{IK}, \mathsf{rsk}_E, \mathsf{ipk}_{IK}, \mathsf{ipk}_{E})$$
then obtains the initial message $\mathsf{msg}_{init}$ and an ephemeral key $\mathsf{rpk}_E$, which enable the IdP to perform the key agreement. The RP forwards the user’s login request, together with $(cid,acid,\pi_D,M)$ to the broker.
Specifically, $cid$ serves as an explicit identifier for the RP in this session, enabling the broker to store and manage future user data associated with the RP.
The broker constructs an authentication request according to standard SSO requirements with $(acid,M)$ and sends it to the IdP.
It retains $\pi_D$ and delivers it to the TCA together with $D$ in the consent phase.
\ret{$\pi_D$ must never reach the IdP, which could otherwise derandomize $acid$
with $t_1$ and match it against every registered RP descriptor to recover the
RP's identity.}

\subsubsection{User Authentication \& Consent}\label{sec:User Authentication and Consent}
\re{The user must verify the RP's identity and determine whether to grant authorization.

The user must first successfully log in to the IdP. Then the IdP loads TCA in the frontend window to help the user verify RP legitimacy.
TCA retrieves $(D,\pi_D)$ from the broker via cross-window messaging, parses $\pi_D=t_1$, and checks
$(0/1) \leftarrow \mathsf{AC.VerifyCred}(\mathsf{ivk}_c,D,acid,\pi_D)$, which derandomizes $acid=(\sigma_1',\sigma_2')$ with $t_1$ and verifies
$\mathsf{PS.Verify}(\mathsf{ivk}_c,(\sigma_1',\sigma_2' \cdot \sigma_1'^{-t_1}),D)=1$.
If successful, $D$ is shown on the consent page. If the user consents to authorize the RP, the IdP proceeds with the following steps.}

\subsubsection{Login Response}
\label{sec:Token Request}
This phase establishes a session key between the RP and IdP, enables anonymous RP authentication, generates an RP-specific user pseudonym, issues authentication credentials and user information $\mathsf{info}$.

\noindent \textbf{Pseudonym Derivation.} The IdP parses $acid=\sigma'=(\sigma_1',\sigma_2')$, where $\sigma_1'=cid^k_1$.
The IdP computes the user’s RP-specific pseudonym for the RP as $auid = \mathsf{BlindEval}(uid, \sigma_1')$

\noindent \textbf{AAKA execution.}  The IdP executes $\mathsf{AnonyAKA}$ to verify proof $\pi_2$:
\[
(0/1) \leftarrow
\mathsf{VerifyProof}(\pi_2, \sigma', \bar{\sigma}', \mathsf{rpk}'_{IK}, (\mathsf{isk}_c,\mathsf{ivk}_c))
\]
confirming the RP’s validity, and derives the shared session key $K_S$.  
The IdP encrypts the user’s sensitive information under the session key $K_S$ as
$\mathsf{info}_{Enc} := \mathsf{Enc}_{K_S}(\mathsf{info})$ and forwards the pair $(auid, \mathsf{info}_{Enc})$ to the RP via the broker.
The RP unblinds the user pseudonym as $uid_{rp} = \mathsf{Unblind}(k_1, auid)$,
decrypts $\mathsf{info}_{Enc}$ with $K_S$, completing the SSO flow.
Finally, the RP delivers the unblinded $uid_{rp}$ to the broker, which records it under $cid$ as the stable per-RP user pseudonym (step~\circnumtwo{11}).
Only the RP holds the blinding factor $k_1$, so this delivery is what gives the broker the $uid_{rp}$ that $\mathsf{BrokerView}$ in Figure~\ref{fig:ideal-orionlink} permits it to observe.

\subsection{Instantiation of \name{}}\label{sec:Instantiation}
We instantiate \name{} based on the widely adopted OIDC Authorization Code flow~\cite{sakimura2014openid}. 

Instantiating OrionLink in practice poses several challenges. In particular, OrionLink must be carefully integrated with the OIDC specifications~\cite{lodderstedt2025rfc}, while preserving the broker’s assistive functionalities described in Section \ref{sec:Brokered SSO Flow}, ensuring that it remains practically efficient and deployable. This fills a critical but often neglected gap in prior privacy-preserving SSO systems \cite{he2024arpsso,heher2025bison,xu2023miso,kroschewski2025oppid}.
To this end, we introduce two features into the OIDC-based instantiation:
(\emph{a}) \emph{\textit{private authorization management:}} allows users to review and revoke their RP accesses without disclosing authorization relationships to the IdP; 
(\emph{b}) \emph{\textit{forward-secure token refresh:}} enables the broker to securely assist RPs in maintaining the validity of user tokens through refresh tokens.

\subsubsection{Private Authorization Management}\label{sec:Private Authorization Management}

While IdP-to-RP anonymous authentication prevents the IdP from tracking users’ authorization choices, users nonetheless require visibility and control over their authorized RPs (e.g., for access review or revocation).  

We address this with the OPRF-based mechanism mentioned in Section \ref{sec:OPRF} that allows the user to
derive a high-entropy key $K_m$ from a short PIN in interaction with the IdP. The derived key $K_m$ is used on the frontend to encrypt the RP identity descriptor $D$,
bind it to the RP pseudonym $acid$, and store the encrypted result $(acid, Enc_{K_m}(D))$ at the IdP.
When the user later reviews or revokes an authorization, the same procedure is executed to re-derive the key and decrypt $D$ for RP identity verification. The corresponding $acid$ is then sent to the IdP to trigger the requested authorization management operation. The details of the algorithm are as follows.

During the \textit{User Authentication \& Consent} phase (step~\circnum{7} in Figure~\ref{fig:overview of OrionLink}), if the user approves the authorization for the RP, a short $PIN$ is entered.
TCA then runs an OPRF with the IdP on behalf of the user to derive a management key $K_m$ from the PIN:  
\begin{enumerate}[leftmargin=*]
    \item TCA computes $r',\alpha=\mathsf{Blind}(\mathsf{PIN})$ and sends $\alpha$ to the IdP;
    \item The IdP evaluates $\beta=\mathsf{BlindEval}(\mathsf{k}_{oprf},\alpha)$ with its OPRF key $\mathsf{k}_{oprf}$ and returns $\beta$;
    \item TCA computes $y=\mathsf{Unblind}(r',\beta)$, derives $K_m=\mathsf{KDF}(y)$, and encrypts the payload as $D_{\mathsf{Enc}}=\mathsf{Enc}_{K_m}(D)$.
\end{enumerate}

This approach derives strong keys from short PINs, reducing user burden while maintaining confidentiality.

\subsubsection{Forward-secure Token Refresh}
\label{sec:Token Refresh}

In the OIDC workflow, the IdP issues three types of tokens: \textit{ID tokens}, \textit{access tokens}, and \textit{refresh tokens} (see details in Appendix~\ref{app:tokens}).
The broker manages \textit{access token} lifecycles on behalf of the RP. 
If an \textit{access token} expires, the broker uses the \textit{refresh token} to request a new one from the IdP.

To preserve privacy against broker surveillance, we design a \emph{key rotation mechanism} inspired by the double ratchet~\cite{doubleratchet}. 
Each refreshed token is encrypted under a freshly derived session key, ensuring forward secrecy. 
Additionally, a Diffie–Hellman ratchet injects new randomness, achieving post-compromise security by preventing past or long-term key exposure from affecting future tokens.
The full protocol details are described in Figure \ref{fig:Token Refresh Protocol with Key Rotation Mechanism} in Appendix \ref{app:Refresh Token Protocol}. 
In practice, the RP pre-generates a sequence of ephemeral DH key pairs and sends the public keys to the broker. 
During operation, the RP simply requests access tokens from the broker. 
If an update is required, the broker uses the next DH key to negotiate a new session key with the IdP. 
The IdP then encrypts the refreshed access token under this key, and the broker returns it to the RP together with the material for updating $K_S$.  

This mechanism allows the broker to continuously manage \textit{access tokens}, 
without requiring RP involvement in every refresh operation, 
while ensuring that only the RP and IdP can access the token contents, preserving both functionality and end-to-end privacy.

\subsection{Security Analysis}\label{sec:security-analysis}

OrionLink is analyzed in the UC framework. The ideal functionality $\mathcal{F}_{\mathsf{OrionLink}}$ in Figure~\ref{fig:ideal-orionlink} captures both the security and the privacy goals, and Theorem~\ref{thm:security} states that the real protocol $\Pi_{\mathsf{OrionLink}}$ securely realizes this functionality. Appendix~\ref{app:security proof} gives the simulation-based proof, formalizes each privacy goal as a standalone definition, and establishes each definition under the assumptions of Theorem~\ref{thm:security} (Corollary~\ref{cor:privacy}).

\begin{theorem}[Security of OrionLink]\label{thm:security} 
Assume that (i) the Decisional Diffie--Hellman (DDH) problem is hard, (ii) the hash functions used in the protocol are modeled as random oracles, so that the $\mathsf{X3DH}$-derived session key is indistinguishable from random, (iii) the symmetric encryption scheme $\mathsf{Enc}$ is IND-CPA secure,
(iv) the anonymous credential scheme satisfies correctness, unforgeability, blindness, and unlinkability, and (v) the DH-based OPRF satisfies correctness, consistency, and pseudorandomness. Then the
protocol $\Pi_{\mathsf{OrionLink}}$ securely realizes the ideal functionality
$\mathcal{F}_{\mathsf{OrionLink}}$.
\end{theorem}
We now spell out what Theorem~\ref{thm:security} implies for each target property of Section~\ref{sec:Target Properties}.
For every property we state how $\mathcal{F}_{\mathsf{OrionLink}}$ models it, which mechanism enforces it in the real protocol, and where its formal definition appears.

\noindent\textbf{Standard SSO Security.} We follow prior work~\cite{fett2017web} in defining standard SSO security and model these properties in $\mathcal{F}_{\mathsf{OrionLink}}$.  An honest IdP issues a token only after a valid user consent and successful authentication of the RP, and the broker forwards the
token exclusively to the corresponding RP session.
We omit an explicit definition of user authentication correctness in $\mathcal{F}_{\mathsf{OrionLink}}$, since our protocol integrates the standard user authentication mechanisms of conventional SSO systems without modification, and inherit the property directly from standard SSO.

\noindent\textbf{Broker Obliviousness.}
\re{This property is captured by the leakage function $\mathsf{BrokerView}$ in Figure~\ref{fig:ideal-orionlink}, which specifies the maximum information revealed to the broker during protocol execution. The broker observes only permitted metadata and the length of transferred user information. In the real protocol, this restriction is enforced by end-to-end encryption between the RP and the IdP (Definition~\ref{def:broker-obliviousness}).}

\noindent\textbf{Untraceability.}
\re{In $\mathcal{F}_{\mathsf{OrionLink}}$, the IdP observes only session-specific identifiers that are independent of the RP's long-term identity. The real protocol realizes this with freshly randomized anonymous credentials and zero-knowledge proofs in each login request, so user--RP authorization relationships remain untraceable at the IdP. In particular, the IdP issues and verifies the RP credential, so the joint issuer--verifier view is already in scope (Definition~\ref{def:untraceability}).}

\noindent\textbf{Unlinkability.}
\re{The RP and the broker observe only RP-specific user pseudonyms, derived via an OPRF as pairwise identifiers bound to the RP's anonymous identifier, so user identities remain unlinkable across RPs even when the broker and multiple RPs pool their views (Definition~\ref{def:unlinkability}).}

\section{Implementation and Evaluation}
\label{sec:Implementation and Evaluation} 
\begin{figure*}[t]
\centering
\captionsetup{justification=centering}
\begin{minipage}[c]{0.34\textwidth}
\centering
\begin{footnotesize}
\captionof{table}{\re{Local computation overhead.}}
\label{tab:local-overhead}

\begin{tabular}{llc}
\toprule
\textbf{Stage} & \textbf{Operation} & \textbf{Time (ms)} \\
\midrule
Initialization & Key generation & 46.59 \\
\midrule
Registration & RP registers to IdP & 5.63 \\
\midrule
\multirow{2}{*}{SSO} & Pseudonym generation & 2.58 \\
 & $\mathsf{AnonyAKA}$ execution & 11.39 \\
\midrule
\multirow{2}{*}{Extension} & Authorization revoke & 0.11 \\
 & Token refresh & 1.20 \\
\midrule
Total &  & 67.50 \\
\bottomrule
\end{tabular}
\end{footnotesize}
\end{minipage}
\hfill
\begin{minipage}[c]{0.31\textwidth}
\centering
\includegraphics[width=\linewidth]{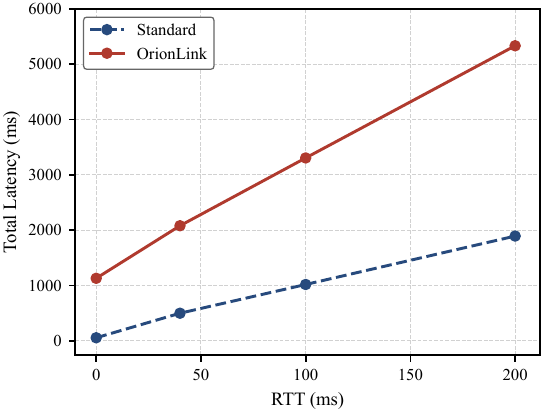}
\captionof{figure}{User login latency\\(\textit{OrionLink} vs. Standard).}
\label{fig:user login latency}
\end{minipage}
\hfill
\begin{minipage}[c]{0.31\textwidth}
\centering
\includegraphics[width=\linewidth]{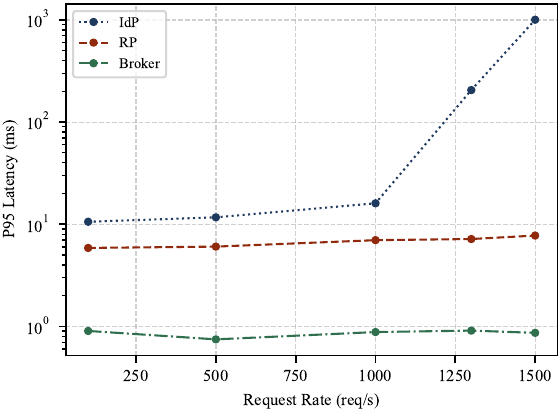}
\captionof{figure}{Stress test of the \\ $\mathsf{AnonyAKA}$ interface.}
\label{fig:plot2}
\end{minipage}

\end{figure*}

\definecolor{lightgray}{gray}{0.9}
\begin{table*}[t]
\centering
\caption{\re{End-to-End performance of \textit{OrionLink} and standard Brokered SSO (RTT=0 ms,40 ms,100 ms,200 ms).}} 
\label{tab:latency-breakdown}
\setlength{\tabcolsep}{4pt}
\begin{footnotesize}
\begin{tabular}{llcccccccc}
\toprule
\multirow{2}{*}{\textbf{Category}}&\multirow{2}{*}{\textbf{Stage}} & \multicolumn{2}{c}{0\,ms} &\multicolumn{2}{c}{40\,ms} & \multicolumn{2}{c}{100\,ms} & \multicolumn{2}{c}{200\,ms} \\
\cmidrule(lr){3-4} \cmidrule(lr){5-6} \cmidrule(lr){7-8} \cmidrule(lr){9-10}
 & 
& \textbf{OrionLink} & \textbf{Standard}
& \textbf{OrionLink} & \textbf{Standard}
& \textbf{OrionLink} & \textbf{Standard}
& \textbf{OrionLink} & \textbf{Standard} \\
\midrule

\multirow{4}{*}{SSO Flow}
& Click ``Login''          & 416.47  &36.69    & 625.67 & 206.60 & 904.23 & 389.51 & 1365.34 & 697.95 \\
& Consent check           &53.01   &--     & 207.28  & --     & 422.99 & --     & 776.94  & --     \\
& Authorization code exchange   &122.27  &19.80 & 595.88 & 293.23 & 1203.45 & 629.36 & 2215.37 & 1195.89 \\
\rowcolor{lightgray}
& \textbf{Login complete (Total) }    &\textbf{1131.08} & \textbf{56.67}      & \textbf{2080.63} & \textbf{500.01} & \textbf{3305.41} & \textbf{1019.12} & \textbf{5331.78} & \textbf{1894.04} \\
\midrule

\multirow{2}{*}{Extension}
& Authorization management   & 20.41  &--   & 68.51  & --     & 130.73 & --     & 235.38  & --     \\
& Token refresh         &2.62    & --       & 123.57 & --    & 272.53 & --    & 521.15 & --     \\
\bottomrule
\end{tabular}
\end{footnotesize}
\end{table*}

This section describes the implementation of the \textit{OrionLink} prototype and presents a comprehensive evaluation of its performance overhead.

\subsection{Implementation}
\label{sec:Implementation}

We implement a prototype of \name{} following the instantiation described in Section~\ref{sec:Instantiation}. 
The prototype comprises over 9,000 lines of code and is publicly available at \url{https://github.com/Ooo0ption/Orionlink.git}.

The system fully realizes all functionalities of the RP, IdP, and broker as specified in the \name{} design. 
Its implementation consists of the following core components.

\noindent(1) \textit{Complete OIDC Login Flow.}
The prototype implements the full OIDC Authorization Code flow, including registration, login request, user authentication and consent, and login response. 
All protocol interactions and request/response constructions strictly follow the OIDC specification~\cite{lodderstedt2025rfc}. 
We adopt the BLS12-381 curve~\cite{bls12} for Pointcheval--Sanders (PS) signatures and issue and verify tokens as JSON Web Tokens (JWTs)~\cite{jones2015jwt}.
All backend components are implemented in Go and leverage Cloudflare’s CIRCL cryptographic library~\cite{circl}.

\vspace{0.05cm}
\noindent(2) \textit{Trusted Computation Agent.}
The TCA is a JavaScript component responsible for executing all user-side cryptographic operations, including anonymous credential verification and OPRF computations. 
It is implemented using the noble-curves~\cite{novel} and facebook-ristretto255~\cite{rist} libraries.

\vspace{0.05cm}
\noindent(3) \textit{Private Authorization Management.}
Private authorization management is realized through user-side OPRF operations executed within the TCA, with corresponding server-side support deployed at the IdP using the Ristretto255 curve~\cite{rfc9496}. 
 
\vspace{0.05cm}
\noindent(4) \textit{Token Refresh with Key Rotation.}
The token refresh mechanism is inspired by the double ratchet protocol~\cite{doubleratchet}, providing forward secrecy across successive refresh operations.

\subsection{Evaluation}
\label{sec:Evaluation}

All RP, broker, and IdP backend components are deployed as Docker containers on a single Ubuntu~24.04 server equipped with a 16-core 4.3\,GHz CPU, 96\,GB of RAM, and a 2\,TB SSD.

\noindent\textbf{Local Computational Overhead}. Table~\ref{tab:local-overhead} reports the local computational overhead of each privacy-enhancing stage in \name{}, with results averaged over 1,000 runs. 
The \textit{RP--broker} registration step is omitted, as it introduces no additional computation. 
The \textit{Authorization Management} measurement captures the time for the TCA to decrypt domain information $D$ on the IdP page during authorization management. 
The \textit{Token Refresh} measurement reports the latency of a single refresh operation involving the RP, broker, and IdP.

Across all stages, the overhead remains minimal, indicating that \name{}'s privacy mechanisms incur negligible local computation cost. 
The total local computational overhead is approximately 67.50~ms.

\noindent\textbf{End-to-End Performance}. 
We evaluate the end-to-end latency of \name{} under realistic user interactions and emulated wide-area network conditions. 
Latency is defined as the elapsed time from when the user clicks the ``Login'' button to the successful completion of authentication at the RP. 

A login session consists of four steps: 
\noindent(1) redirecting to the IdP consent page after the user clicks ``Login''; 
\noindent(2) user consent with frontend execution of the TCA; 
\noindent(3) authorization code exchange by the RP; 
\noindent(4) successful login completion.

We implement a standard brokered SSO baseline using the OIDC Authorization Code flow.
Except for additional UI pages required by \textit{OrionLink}, both systems follow similar webpage structures and interaction patterns. 
User interactions are automated using Playwright~\cite{playwright} to deterministically trigger click events. 
Table~\ref{tab:latency-breakdown} reports the results under emulated round-trip time (RTT) network latencies of 0~ms, 40~ms, 100~ms, and 200~ms, with each result averaged over 100 runs.

Compared to the baseline, \textit{OrionLink} introduces additional latency due to its privacy-enhancing mechanisms and the resulting increase in user interactions, page redirections, and frontend rendering overhead. 
These overheads primarily stem from additional login-phase redirections, frontend execution of the TCA during consent, and cross-window communication during authorization code exchange. 
As shown in Figure~\ref{fig:user login latency}, the end-to-end latency of both systems increases approximately linearly with network delay.
Under a typical network latency of 40~ms, a complete \textit{OrionLink} authentication takes 2.08~s on average, with the privacy extensions contributing an additional 1.58~s.
The overall login latency remains acceptable for interactive use.
Future optimizations could further reduce latency by improving frontend efficiency (e.g., through WebAssembly~\cite{zhang2020passo}) and streamlining the interaction flow to reduce protocol rounds.

\noindent\textbf{Stress Test.}
Local overhead analysis identifies AAKA execution as the most computationally intensive component of \textit{OrionLink}. 
We therefore conduct stress tests on the AAKA interfaces jointly executed by the RP, broker, and IdP. 
As shown in Figure~\ref{fig:plot2}, under a load of 1,000 requests per second (RPS), the system maintains a P95 latency of approximately 10~ms, indicating stable and low-latency performance even under high load. 

\noindent\textbf{Mobile Evaluation.}
\re{The TCA executes only a fixed set of cryptographic operations in the browser during one login. This set takes 32.4\,ms on an iPhone~15 (iOS~18.7, Safari) versus 32.2\,ms on a consumer laptop (AMD Ryzen~7 5800HS, 16\,GB RAM, Edge), both medians over 100 runs, so mobile clients incur no measurable extra cost.}

\section{Discussion} \label{sec:discussion}

\begin{table*}[t]
\begin{threeparttable}
\centering
\begin{footnotesize}
\caption{Comparison of \textit{OrionLink} with prior privacy-preserving SSO schemes.}
\label{tab:comparison-feature}
\renewcommand\theadfont{\bfseries}
\renewcommand\theadalign{cc}
\renewcommand\theadset{\renewcommand\arraystretch{0.5}}
\setlength{\arrayrulewidth}{0.4pt}
\newcommand{\halfcircle}{%
    \begin{minipage}{0.03\linewidth}\centering
    \tikz[baseline=-0.5ex] {
        \fill[black] (90:0.1cm) arc (90:-90:0.1cm) -- cycle;
        \draw (90:0.1cm) arc (90:270:0.1cm);
    }%
    \end{minipage}
}
\begin{tabular}{c|ccc|cc|cc}
\toprule
\multirow{2}{*}{\thead{Scheme}} &
\multicolumn{3}{c|}{\thead{Privacy}} &
\multicolumn{2}{c|}{\thead{Brokered SSO Features}} &
\multicolumn{2}{c}{\thead{Extension}}\\
\cmidrule(lr){2-4} \cmidrule(lr){5-6} \cmidrule(lr){7-8}
 & \makecell[c]{Identity \\ Obliviousness} 
 & \makecell[c]{Credential \\ Obliviousness} 
  & \makecell[c]{Personal Info. \\ Obliviousness} 
 & \makecell[c]{End-to-End \\ Confidentiality} 
 & \makecell[c]{Broker \\ Compatibility} 
  & \makecell[c]{ Authorization \\ Mgmt.} 
    & \makecell[c]{Forward-Secure\\Token Refresh} \\
\hline
UP-SSO \cite{guo2022up}  & \cmark & \xmark & \xmark & \xmark & \xmark & \xmark & \xmark\\
BISON \cite{heher2025bison}    & \cmark & \xmark & \xmark & \xmark & \xmark & \xmark & \xmark\\
AIF-ZKP \cite{kroschewski2023save} & \halfcircle\tnote{1} & \xmark & \xmark & \xmark & \xmark & \halfcircle\tnote{2} & \xmark\\
ARPSSO \cite{he2024arpsso}  & \cmark & \xmark & \xmark & \xmark & \xmark & \xmark & \xmark \\
MISO \cite{xu2023miso}    & \cmark & \cmark & \cmark & \xmark & \xmark & \xmark & \xmark\\
\hline
\textbf{\textit{OrionLink}}  & \cmark & \cmark & \cmark & \cmark & \cmark & \cmark & \cmark\\
\bottomrule
\end{tabular}
\end{footnotesize}
\begin{tablenotes}
\footnotesize
\item[] \textsuperscript{1}Unlinkability limited to the User--RP relationship. \quad \textsuperscript{2}RP authorization control only indirectly, through validity-period enforcement.
\end{tablenotes}
\end{threeparttable}
\end{table*}

\noindent\textbf{Choice of PS over BBS+.}
\ret{\name{} builds on one structural property of PS signatures. The first
component $\sigma_1=h$ is a random group element chosen by the issuer, reused in
\name{} as the per-RP identifier $cid:=h$. Re-randomizing the credential gives
$\sigma_1'=h^{k_1}$, from which the IdP blindly evaluates
$auid=(h^{k_1})^{uid}=(h^{uid})^{k_1}$ and the RP unblinds $uid_{rp}=h^{uid}$, so
the RP credential and the OPRF-based pseudonym derivation of
Section~\ref{sec:Design} share one algebraic base. BBS+ does not expose an analogous independently chosen group-element component that can be directly reused as the pseudonym base, and would therefore require a separate pseudonym-derivation mechanism.}

\noindent\textbf{Compatibility with Other SSO Protocols.}
\textit{OrionLink} is compatible with other federated SSO protocols, including SAML~2.0~\cite{hughes2005security} and WS-Federation~\cite{standardweb}. 
At a high level, these protocols follow a common authentication workflow, as described in Section~\ref{sec:Brokered SSO Flow}, upon which OrionLink is built. 
Differences are largely confined to message formats and token representations and do not fundamentally affect OrionLink’s design. 
That said, the privacy guarantees targeted by OrionLink are not required in all SSO scenarios. 
In enterprise settings using SAML~2.0 or WS-Federation, RPs and IdPs are often under common organizational control and already have visibility into users’ identities and access patterns, making anonymity less critical. 
In such environments, anonymous credentials may be unnecessary, while OrionLink’s end-to-end encryption can still provide value by protecting authentication data from broker observation.

\noindent \textbf{TCA Security.}
In OrionLink's deployment, the TCA runs in a cross-origin \texttt{<iframe>} on an origin independent of the IdP, the broker, and any RP, so the same-origin policy prevents IdP-hosted scripts from reading TCA-internal state such as the PIN, $K_m$, or plaintext $D$. The TCA page loads its open, version-pinned script bundle via an SRI-pinned \texttt{<script>} element, while versioned URLs and an independently auditable transparency log provide bundle provenance. We defer the remaining TCA security details to Appendix~\ref{app:TCA}.

\noindent \textbf{Web Attacks and Side Channels.}
\re{OrionLink builds an end-to-end encrypted channel between the RP and the IdP through AAKA, which resists the network attacks that arise once the broker is inserted as an additional relay point. In standard brokered SSO, a mis-validated redirection chain delivers plaintext authentication tokens to an attacker-chosen endpoint (Appendix~\ref{app:Measurements of Vulnerabilities}), whereas in OrionLink the same attack yields only $\mathsf{info}_{Enc}$ and content-free pseudonyms that reveal nothing without the session key $K_S$ shared by the RP and the IdP. The brokered architecture also removes one class of side channels against the RP, since every login request reaches the IdP from the broker, so the IdP observes neither the request fingerprint nor the network traffic of the RP. Side channels beyond the model remain, including the timing of TCA computations and correlation by a global passive adversary, and established countermeasures apply here as well, such as building the TCA from constant-time cryptographic components. A compromised browser or malicious extension instead breaks the client-side trust anchor altogether, and we discuss this case in Appendix~\ref{app:TCA}.}

\section{Related Works}
We review prior work on privacy-preserving SSO and position \textit{OrionLink} within research on brokered authentication.

\noindent\textbf{Privacy-preserving SSO.} 
To the best of our knowledge, \textit{OrionLink} is the first privacy-enhancing framework designed for \emph{brokered} SSO. In contrast, most existing approaches~\cite{hammann2020privacy,fett2015spresso,heher2025bison,xu2023miso,kroschewski2025oppid} focus on standard SSO without an intermediary broker. Table~\ref{tab:comparison-feature} summarizes a detailed comparison. 
Several schemes~\cite{lin2025understanding,guo2022up,heher2025bison,xu2023miso,kroschewski2025oppid} employ anonymous credentials to achieve properties such as \emph{user behavior untraceability} and \emph{cross-RP unlinkability}, which \textit{OrionLink} also provides. However, many rely on outdated or insecure protocol designs. In particular, multiple proposals~\cite{fett2015spresso,hammann2020privacy,kroschewski2023save} are based on the deprecated OIDC Implicit Flow, which lacks RP authentication and is vulnerable to token leakage. By contrast, ARPSSO~\cite{he2024arpsso} was the first system built entirely on the modern OIDC Authorization Code flow, a secure-by-design approach that \textit{OrionLink} also adopts. 
\textit{OrionLink} uniquely supports key brokered SSO properties, including \emph{(i) end-to-end confidentiality}, preventing brokers from accessing tokens or sensitive user data, and \emph{(ii) broker compatibility}, largely preserving the assistive functions critical to existing brokered SSO ecosystems. Although MISO~\cite{xu2023miso} was designed for standard SSO, its TEE-based mixer intermediates communication between the RP and IdP to hide credentials and personal data, conceptually resembling a broker. However, reliance on TEEs introduces additional trust assumptions and deployment complexity.  
In AIF-ZKP~\cite{kroschewski2023save}, RP authorization is enforced by limiting the lifetime of RP credentials used in anonymous authentication. This approach remains passive and does not support active, user-driven authorization management.

\noindent\textbf{Brokered SSO in Emerging Architectures.} 
The importance of secure brokered SSO has increased with emerging architectures such as MCP, an open standard for connecting AI systems with external applications~\cite{hou2025model,mcpauth,cloudflare}. 
In MCP, servers simultaneously act as IdPs for clients and as RPs for upstream services, naturally fitting the brokered SSO model and making delegated authorization essential. 
Cloudflare~\cite{cloudflare}, a widely deployed reverse proxy, promotes brokered SSO as a practical means of integrating authorization into MCP, recommending ``\emph{authorization-as-a-service}'' providers such as Stytch~\cite{stytch}, Auth0~\cite{auth0}, and WorkOS~\cite{workos}. As brokered SSO adoption grows in AI-driven ecosystems, privacy and data exposure risks increase accordingly, highlighting the need for systems like \textit{OrionLink}.
\vspace{-0.2cm}
\section{Conclusion}

In this work, we identify data exposure at identity brokers as a fundamental security and privacy risk in brokered SSO and propose \textbf{OrionLink}, 
the first privacy-preserving framework designed for \textit{brokered SSO} to mitigate risks posed by \textit{identity brokers}. 
OrionLink adopts \emph{broker obliviousness} as its central design principle and minimizes the information observable by the broker throughout protocol execution. 
Specifically, OrionLink introduces (\emph{a}) an anonymous pseudonym derivation scheme that generates unlinkable, RP-specific user pseudonyms and session-bound RP pseudonyms, preventing cross-service behavioral tracking, 
and (\emph{b}) an anonymous key agreement protocol that establishes end-to-end encrypted channels between RPs and IdPs, shielding users’ authentication and authorization data from broker visibility. 
We further instantiate OrionLink in a practical system and incorporate two features to enhance usability and security.

We formally prove the security of OrionLink, and our prototype implementation and experimental evaluation demonstrate that the system achieves all stated design goals with minimal performance overhead, 
while providing strong privacy guarantees and a seamless user experience.  
In particular, OrionLink attains an end-to-end login latency of about 2.1 seconds in a WAN setting, underscoring its practicality for deployment in real-world brokered SSO ecosystems.

\subsection{Ethical Considerations}
This work presents a defensive, privacy-enhancing framework that does not involve human subjects, real user data, or attacks on deployed systems.

\section*{Acknowledgment}
This work was supported in part by the Beijing Natural Science Foundation under Grant 4242023; in part by the National Natural Science Foundation of China under Grant 62472021; and in part by the Fundamental Research Funds for the Central Universities.

\bibliographystyle{IEEEtran}
\bibliography{Article/cite}

\appendices

\section{Open Science}
To support transparency and reproducibility, we provide a prototype implementation of our system. The source code is publicly available at \url{https://github.com/Ooo0ption/Orionlink.git}.

Our prototype system fully implements all functionalities of the RP, IdP, and broker as designed in OrionLink. The core artifacts included are described in Section \ref{sec:Implementation}.

\section{Brokered SSO Ecosystem}
We briefly summarize the current Brokered SSO ecosystem, with key measurements primarily referenced from a recent study \cite{innocenti2025only}. More details and results can be found in the original paper.
\subsection{Measurements of Ecosystem.}\label{app:Measurements of Ecosystem}
Based on publicly available data, the study \cite{innocenti2025only} filtered 55,083 SSO flows considered valid for analysis. As shown in Table \ref{tab:brokers-ecosystem}, 249 Brokers were identified in 15,213 SSO flows, with approximately 25\% of websites integrating Brokered SSO for authentication.
\begin{table}[t]
\begin{footnotesize}
  \caption{The Brokered SSO Ecosystem}
  \label{tab:brokers-ecosystem}
   \setlength{\tabcolsep}{3pt}
  \centering
  \begin{tabular}{cccccccc}
    \toprule
    \multicolumn{1}{c}{\textbf{Model}}& &
    \multicolumn{2}{c}{\textbf{Brokers}} &
    \multicolumn{2}{c}{\textbf{Flows}} &
    \multicolumn{2}{c}{\textbf{Websites}} 
    \\
    \hline
   \textbf{Three-Party}& & - & - & \textbf{39,870 }& \textbf{72\%} & \textbf{24,880} &\textbf{75\%}\\
    \hline
  \multirow{4}{*}{\textbf{Brokered}} & All Types & \textbf{249} & \textbf{100\%} & \textbf{15,213} & \textbf{28\%} & \textbf{8,241} & \textbf{25\%}  \\
\cmidrule(lr){2-8}
\quad & \multicolumn{1}{l}{\# Public}   & 64  & 26\% & 10,911 & 71\% & 5,944 & 72\%  \\
\quad & \multicolumn{1}{l}{\# Internal} & 117 & 47\% & 3,442  & 22\% & 1,635 & 20\%  \\
\quad & \multicolumn{1}{l}{\# Unknown}  & 68  & 27\% & 1,139  &  7\% &   686 &  8\%  \\
    \midrule
   \textbf{ Total } &  &  & &55,083 & 100\% & 33,121 & 100\% \\
    \bottomrule
  \end{tabular}
\end{footnotesize}
\end{table}

\textit{Public vs. Internal Brokers.}
Public brokers are open to any RP and may be offered as Software as a Service (SaaS) solutions or as open-source options for self-hosted deployments. 
By contrast, internal brokers serve only a specific set of RPs, often within the same organization such as government agencies. 

As shown in Table \ref{tab:brokers-ecosystem}, public brokers have a broader impact on the overall brokered SSO ecosystem. We therefore focus more on the privacy and vulnerability issues introduced by public brokers.

\subsection{Measurements of Vulnerabilities.} \label{app:Measurements of Vulnerabilities}
 The prior work \cite{innocenti2025only} identified three major vulnerabilities currently faced by brokered SSO: 
\begin{itemize}[leftmargin=*]
    \item \textit{Open redirection chains\cite{ORS}.} About 20\% (49 of 249) of brokers fail to sufficiently validate redirection chains, leading to the leakage of authentication tokens and, in turn, enabling account takeover attacks. (see Table \ref{tab:brokers-vulnerability})
    \item \textit{Unauthorized Data Access.} Some brokers lack strict access control and authorization validation mechanisms, allowing attackers to access user data without authorization. This results in data leakage and account overtake attacks. (see Table \ref{tab:brokers-vulnerability})
    \item \textit{Violations of SBCPs.} Some brokers (with 3,367 violations observed out of 5,668 flows) fail to adhere to the recommended specifications, thereby weakening the security guarantees of SSO flows.
\end{itemize}

\begin{table}[t]
\begin{footnotesize}
  \caption{Vulnerabilities in Brokered SSO Ecosystem}
  \label{tab:brokers-vulnerability}
  \centering
  
  \begin{tabular}{cccc}
    \toprule
    \textbf{Vulnerabilities} &
   \textbf{Impact} &
    \textbf{Brokers} &
   \textbf{Affected RPs}
    \\
    \hline
  \multirow{4}{1.8cm}{\textbf{Open Redirection Chains}} & \multicolumn{1}{l}{Data Leakage}   & 38  & 1,979 \\
 & \multicolumn{1}{l}{Account Takeover} & 34 & 1,950\\
\cmidrule(lr){2-4}
& \multicolumn{1}{l}{\textbf{Total}}  & \textbf{49} & \textbf{2,421}\\
\hline
  \multirow{4}{1.8cm}{\textbf{Unauthorized Data Access}} & \multicolumn{1}{l}{Data Leakage} & 15 & 1,103   \\
 & \multicolumn{1}{l}{Account Takeover}   & 6 & 779 \\
 \cmidrule(lr){2-4}
 & \multicolumn{1}{l}{ \textbf{Total} } & \textbf{18} & \textbf{1,149}\\
    \bottomrule
  \end{tabular}
\end{footnotesize}
\end{table}
\section{OIDC Tokens} \label{app:tokens}
OIDC issues three main types of tokens to enable authentication and authorization \cite{sakimura2014openid}:
\begin{itemize}[leftmargin=*]
    \item \textit{ID Token}: A signed JSON Web Token issued by the IdP after user authentication. It conveys user identity claims and includes fields such as \texttt{iss} for the issuer, \texttt{sub} for the user identifier, and \texttt{aud} for the intended audience or client.  
   \item \textit{Access Token}: A token issued by the IdP that allows the client to access protected resources or APIs on behalf of the user. It carries authorization information such as scopes and resource identifiers but does not expose full user identity details. 
  \item \textit{Refresh Token}: A long-lived token issued together with the access token. It enables the client to obtain new access tokens without requiring the user to re-authenticate.
\end{itemize}

\section{More Details of \textit{OrionLink}} \label{app:signal}

\begin{figure}[t]
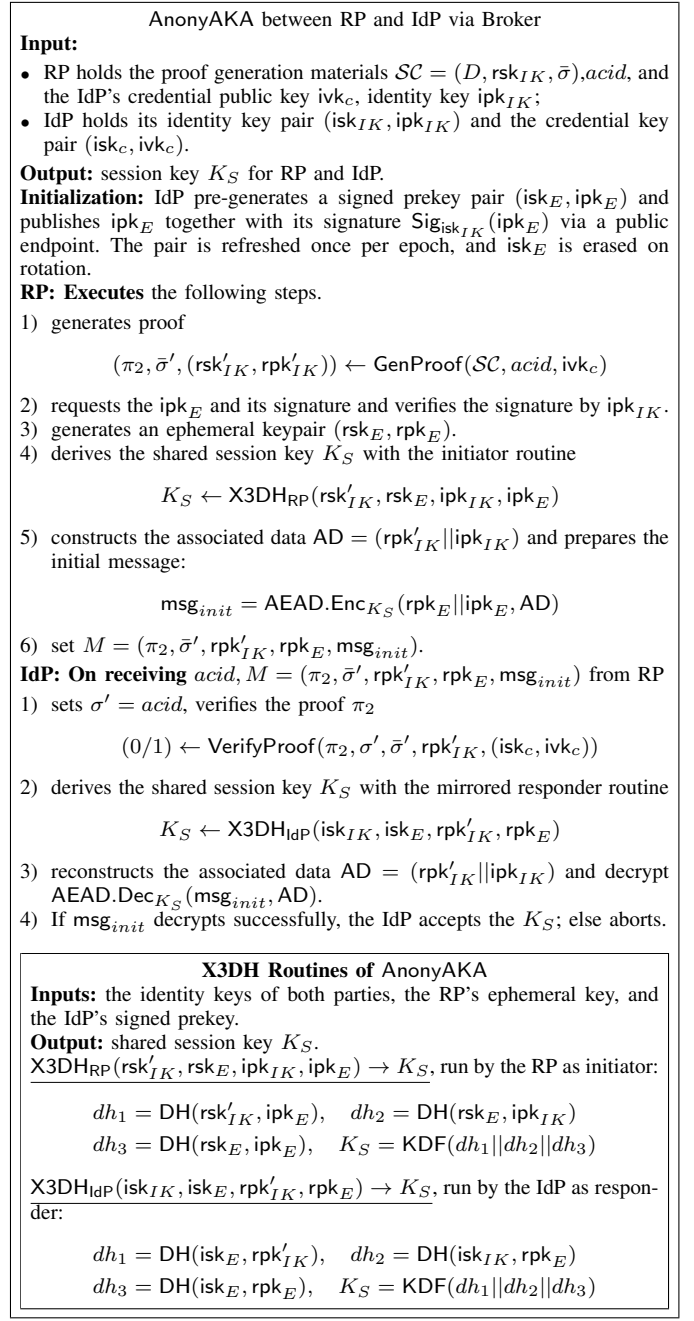

\begin{mdframed}[innerleftmargin=3.5pt, innerrightmargin=3.5pt, innertopmargin=3pt, innerbottommargin=3pt]
\begin{footnotesize}
\begin{center}
    $\mathsf{AnonyAKA}$ between RP and IdP via Broker
\end{center}
\textbf{Input:}
\begin{itemize}[leftmargin=*]
  \item RP holds the proof generation materials $\mathcal{SC}=(D,\mathsf{rsk}_{IK},\bar{\sigma})$,$acid$, and the IdP’s credential public key $\mathsf{ivk}_c$, identity key $\mathsf{ipk}_{IK}$;
  \item IdP holds its identity key pair $(\mathsf{isk}_{IK},\mathsf{ipk}_{IK})$ and the credential key pair $(\mathsf{isk}_c,\mathsf{ivk}_c)$.
\end{itemize}
\textbf{Output:} session key $K_S$ for RP and IdP.

\textbf{Initialization:}
IdP pre-generates a signed prekey pair $(\mathsf{isk}_{E},\mathsf{ipk}_{E})$ and publishes $\mathsf{ipk}_{E}$ together with its signature $\mathsf{Sig}_{\mathsf{isk}_{IK}}(\mathsf{ipk}_{E})$ via a public endpoint. \ret{The pair is refreshed once per epoch, and $\mathsf{isk}_{E}$ is erased on rotation.}

\textbf{RP: Executes} the following steps.
\begin{enumerate}[leftmargin=*]
  \item generates proof 
  $$(\pi_2, \bar{\sigma}', (\mathsf{rsk}'_{IK},\mathsf{rpk}'_{IK})) \leftarrow \mathsf{GenProof}(\mathcal{SC},acid,\mathsf{ivk}_c)$$
  \item requests the $\mathsf{ipk}_{E}$ and its signature and verifies the signature by $\mathsf{ipk}_{IK}$.
  \item generates an ephemeral keypair $(\mathsf{rsk}_E, \mathsf{rpk}_E)$.
  \item derives the shared session key $K_S$ with the initiator routine
  \begin{align*}
     K_S \leftarrow \mathsf{X3DH}_{\mathsf{RP}}(\mathsf{rsk}'_{IK}, \mathsf{rsk}_E, \mathsf{ipk}_{IK}, \mathsf{ipk}_{E})
  \end{align*}
  \item constructs the associated data $\mathsf{AD} = (\mathsf{rpk}'_{IK}|| \mathsf{ipk}_{IK})$ and prepares the initial message:
   \begin{align*}
    \mathsf{msg}_{init} = \mathsf{AEAD.Enc}_{K_S}(\mathsf{rpk}_E||\mathsf{ipk}_{E},\mathsf{AD}) 
  \end{align*}
  \item set $M=(\pi_2, \bar{\sigma}', \mathsf{rpk}'_{IK}, \mathsf{rpk}_E,\mathsf{msg}_{init})$.
\end{enumerate}

 \textbf{IdP: On receiving} $acid, M=(\pi_2, \bar{\sigma}', \mathsf{rpk}'_{IK}, \mathsf{rpk}_E,\mathsf{msg}_{init})$ from RP
\begin{enumerate}[leftmargin=*]
  \item sets $\sigma'=acid$, verifies the proof $\pi_2$
  $$(0/1) \leftarrow \mathsf{VerifyProof}(\pi_2,\sigma', \bar{\sigma}',\mathsf{rpk}'_{IK},(\mathsf{isk}_c,\mathsf{ivk}_c))$$
  \item derives the shared session key $K_S$ with the mirrored responder routine
  \begin{align*}
    K_S \leftarrow \mathsf{X3DH}_{\mathsf{IdP}}(\mathsf{isk}_{IK}, \mathsf{isk}_{E}, \mathsf{rpk}'_{IK}, \mathsf{rpk}_{E})
  \end{align*}
  \item reconstructs the associated data $\mathsf{AD} = (\mathsf{rpk}'_{IK}|| \mathsf{ipk}_{IK})$ and decrypt $ \mathsf{AEAD.Dec}_{K_S}(\mathsf{msg}_{init},\mathsf{AD})$.
  \item If $\mathsf{msg}_{init}$ decrypts successfully, the IdP accepts the $K_S$; else aborts.
\end{enumerate}
\begin{mdframed}
\reon
\begin{center}
    \textbf{X3DH Routines of $\mathsf{AnonyAKA}$}
\end{center}
\textbf{Inputs:} \ret{the identity keys of both parties, the RP's ephemeral key, and the IdP's signed prekey.}

\textbf{Output:} shared session key $K_S$.

\noindent $\underline{\mathsf{X3DH}_{\mathsf{RP}}(\mathsf{rsk}'_{IK}, \mathsf{rsk}_E, \mathsf{ipk}_{IK}, \mathsf{ipk}_{E}) \rightarrow K_S}$, run by the RP as initiator:
\begin{align*}
dh_1 &= \mathsf{DH}(\mathsf{rsk}'_{IK}, \mathsf{ipk}_{E}),\quad
dh_2 = \mathsf{DH}(\mathsf{rsk}_E, \mathsf{ipk}_{IK})\\
dh_3 &= \mathsf{DH}(\mathsf{rsk}_E, \mathsf{ipk}_{E}), \quad
K_S = \mathsf{KDF}(dh_1||dh_2||dh_3)
\end{align*}

\noindent $\underline{\mathsf{X3DH}_{\mathsf{IdP}}(\mathsf{isk}_{IK}, \mathsf{isk}_{E}, \mathsf{rpk}'_{IK}, \mathsf{rpk}_{E}) \rightarrow K_S}$, run by the IdP as responder:
\begin{align*}
dh_1 &= \mathsf{DH}(\mathsf{isk}_{E}, \mathsf{rpk}'_{IK}),\quad
dh_2 = \mathsf{DH}(\mathsf{isk}_{IK}, \mathsf{rpk}_{E})\\
dh_3 &= \mathsf{DH}(\mathsf{isk}_{E}, \mathsf{rpk}_{E}), \quad
K_S = \mathsf{KDF}(dh_1||dh_2||dh_3)
\end{align*}
\end{mdframed}
\end{footnotesize}
\end{mdframed}
\caption{Anonymous Authentication Key Agreement.} %
  \label{fig:Anonymous Authentication Key Exchange Protocol}
\end{figure}

\noindent\textbf{Notations.}
The notations used in the Token Refresh Protocol are defined as follows.
(i) $\text{RK}$, the root key, initialized from the session key $K_S$ established via the AAKA protocol. The ratchet state $(\text{RK},\text{CK})$ is held only by the RP and the IdP, the two parties that know $K_S$, while the broker keeps no key state and only relays headers and ciphertexts.
(ii) $\text{CK}$, the chain key, used by the IdP as a sending chain to derive symmetric message keys.
(iii) $\mathsf{KDF\_RK}(\text{RK}, dh)$, a $\mathsf{KDF}$ that updates the root key and produces a new chain key from the current root key and a fresh DH shared secret $dh$.
(iv) $\mathsf{KDF\_CK}(\text{CK})$, a $\mathsf{KDF}$ that derives the next chain key and a message key $\mathsf{MK}$ from the current chain key.

\noindent\textbf{AAKA Protocol $\mathsf{AnonyAKA}$}\label{app:AAKE Protocol}
This protocol is described in Sections~\ref{sec:Authorization Request} and~\ref{sec:Token Request}. Figure~\ref{fig:Anonymous Authentication Key Exchange Protocol} shows the Anonymous Authentication Key Agreement (AAKA) protocol, which is an end-to-end session key establishment protocol that allows the IdP to verify RP legitimacy under anonymity and derive a unique session key for each authorization process.

\noindent\textbf{Proof Generation and Verification.}\label{app:Proof Generation and Verification}
The algorithms are invoked in the protocol flow described in Section \ref{sec:Token Request}.
Figure \ref{fig:A Proof Generation Algorithm in AAKE Protocol} shows (1) $\mathsf{GenProof}$, that produces a randomized signature $\bar{\sigma}'$ from $\mathcal{SC}$, an ephemeral public identity key $\mathsf{rpk}'_{IK}$ for the AAKA protocol, and a NIZK proof $\pi_2$; (2)$\mathsf{VerifyProof}$, that allows the IdP to check the validity of $\pi_2$.  The NIZK proof $\pi_2$ is obtained by applying the Fiat--Shamir transform to a
standard $\Sigma$-protocol, and its zero-knowledge and soundness properties
follow directly in the random oracle model. 

\noindent\textbf{Token Refresh Protocol $\mathsf{RefToken}$.} \label{app:Refresh Token Protocol}
This protocol is described in Section~\ref{sec:Token Refresh}. Figure~\ref{fig:Token Refresh Protocol with Key Rotation Mechanism} shows a token refresh protocol that employs a double-ratchet–based key rotation mechanism to protect token confidentiality from broker surveillance.

\begin{figure}[t]
\begin{mdframed}[innerleftmargin=3.5pt, innerrightmargin=3.5pt, innertopmargin=3pt, innerbottommargin=3pt]
\begin{footnotesize}
\begin{flushleft}
$\underline{(\pi_2, \bar{\sigma}', (\mathsf{rsk}'_{IK},\mathsf{rpk}'_{IK})) \leftarrow \mathsf{GenProof}(\mathcal{SC},acid,\mathsf{ivk}_c)}$

\noindent \textbf{RP inputs:}
\begin{itemize}[leftmargin=*]
    \item the secret credential $\mathcal{SC}=(D,\mathsf{rsk}_{IK}, \bar{\sigma})$.
    \item $acid=\sigma' = (\sigma'_1, \sigma'_2)$ with randomness $k_1, t_1$
    \item credential verification key $\mathsf{ivk}_c$.
\end{itemize}
\textbf{Output:} a NIZK proof $\pi_2=(\vec{C},\vec{R},\vec{z})$, an ephemeral identity keypair $(\mathsf{rsk}'_{IK},\mathsf{rpk}'_{IK})$, randomized signature $\bar{\sigma}'$.

\noindent \textbf{RP executes the following steps:}
\begin{enumerate}[leftmargin=*]
    \item Randomize signature $\bar{\sigma}$ in $\mathcal{SC}$:
\begin{align*}
k_2,t_2 \narrowxleftarrow{\$} \mathbb{Z}^*_p, \bar{\sigma}'=(\bar{\sigma}_1^{'},\bar{\sigma}_2^{'})=(\bar{\sigma}_1^{{k_2}},(\bar{\sigma}_2 \bar{\sigma}_1^{{t_2}})^{k_2})
\end{align*}
\item Generate an ephemeral public identity key:
\begin{align*}
t_3 \narrowxleftarrow{\$} \mathbb{Z}^*_p, \mathsf{rsk}'_{IK} = t_3 + \mathsf{rsk}_{IK},  \mathsf{rpk}'_{IK} = g_1^{t_3 + \mathsf{rsk}_{IK}}    
\end{align*}
\item Sample randomness: $r_1, r_2, r_3, r_D,r_{rsk}  \narrowxleftarrow{\$} \mathbb{Z}^*_p$
\item Compute commitment $\vec{C}=(C_1, C_2, C_3)$  and response commitment $\vec{R} = (R_1, R_2, R_3)$
\begin{align*}
  C_1 &= g_2^{t_1} \cdot \bar{Y}_1^D, \quad & R_1 &= g_2^{r_1} \cdot \bar{Y}_1^{r_D} \\
    C_2 &= g_2^{t_2} \cdot \bar{Y}_{21}^{\mathsf{rsk}_{IK}} \cdot \bar{Y}_{22}^D, \quad & R_2 &= g_2^{r_2} \cdot \bar{Y}_{21}^{r_{\mathsf{rsk}}} \cdot \bar{Y}_{22}^{r_D} \\
    C_3 &= g_1^{t_3} \cdot g_1^{\mathsf{rsk}_{IK}}, \quad & R_3 &= g_1^{r_3} \cdot g_1^{r_{\mathsf{rsk}}}
\end{align*}
\item Compute challenge $c = H(C_1||C_2||C_3||R_1||R_2||R_3)$
\item Compute response vector $\vec{z} = (z_1, z_2, z_3, z_D, z_{\mathsf{rsk}})$
  \begin{align*}
    z_1 &= r_1 + c \cdot t_1, \quad
    z_2 = r_2 + c \cdot t_2, \quad
    z_3 = r_3 + c \cdot t_3 \\
    z_D &= r_D + c \cdot D, \quad
    z_{\mathsf{rsk}} = r_{\mathsf{rsk}} + c \cdot \mathsf{rsk}_{IK}
  \end{align*}
\end{enumerate}

\end{flushleft}
\begin{flushleft}
$\underline{ (0/1) \leftarrow \mathsf{VerifyProof}(\pi_2,\sigma', \bar{\sigma}',\mathsf{rpk}'_{IK},(\mathsf{isk}_c,\mathsf{ivk}_c))}$

\noindent \textbf{IdP inputs:}
\begin{itemize}[leftmargin=*]
    \item two signatures $\sigma', \bar{\sigma}'$ and NIZK proof $\pi_2$
    \item credential keypair $(\mathsf{isk}_c,\mathsf{ivk}_c)$ and public key $\mathsf{rpk}'_{IK}$.
\end{itemize}
\textbf{Output:} accept iff all checks pass; otherwise reject.

\noindent \textbf{IdP executes the following steps:}

\begin{enumerate}[leftmargin=*]
    \item extract $\vec{C},\vec{R},\vec{z}$ from $\pi_2$;
    \item re-compute the challenge $c = H(C_1||C_2||C_3||R_1||R_2||R_3)$
    \item check the following equations:
 \begin{align*}
     C_1^c \cdot R_1 &\stackrel{?}{=} g_2^{z_1} \cdot \bar{Y}_1^{z_D}\\ \quad C_2^c \cdot R_2 &\stackrel{?}{=} g_2^{z_2} \cdot \bar{Y}_{21}^{z_{\mathsf{rsk}}} \cdot \bar{Y}_{22}^{z_D} \\
 C_3^c \cdot R_3 &\stackrel{?}{=} g_1^{z_3} \cdot g_1^{z_{\mathsf{rsk}}}
 \end{align*}
    \item check $C_3 = \mathsf{rpk}'_{IK}$ and the two pairing equations:
\[ e(\sigma_{1}', C_1 \cdot \bar{X}_1) \stackrel{?}{=} e(\sigma_{2}', g_2), \quad e(\bar{\sigma}_{1}', C_2 \cdot \bar{X}_2) \stackrel{?}{=} e(\bar{\sigma}_{2}', g_2)\]
\end{enumerate}
\end{flushleft}
\end{footnotesize}
\end{mdframed}
\caption{The Proof Algorithm in AAKA Protocol} %
  \label{fig:A Proof Generation Algorithm in AAKE Protocol}
\end{figure}
\begin{figure}[t]
\begin{mdframed}[innerleftmargin=3.5pt, innerrightmargin=3.5pt, innertopmargin=3pt, innerbottommargin=3pt]
\begin{footnotesize}
\begin{center}
$\mathsf{RefToken}$ between RP, IdP and Broker
\end{center}

\textbf{Input:} 
\begin{itemize}[leftmargin=*]
    \item RP and IdP hold the session key $K_S$ established from AAKA protocol.
    \item Broker holds the $\mathsf{refresh\_token}$
\end{itemize}
\textbf{Output:} a new $\mathsf{access\_token}$ for RP.

\textbf{\re{Initialization:}}
\begin{itemize}[leftmargin=*]
\reon
\item RP sets its initial state:\\
-- initializes $\text{RK} \leftarrow K_S$, $\text{CK} \leftarrow \bot$\\
-- pre-computes a list of DH key pairs $\{(\mathsf{rsk}_i^{dh}, \mathsf{rpk}_i^{dh})\}_{i=0}^{n-1}$\\
-- sets current DH key state: $\text{DH}_A^{rp} \leftarrow \bot $, $\text{DH}_B^{rp} \leftarrow \bot$
\item Broker sets its initial state:\\
  -- initializes $seq \leftarrow 0$\\
  -- loads $\{\mathsf{rpk}_i^{dh}\}_{i=0}^{n-1}$ from RP\\
  -- sets current DH key state: $\text{DH}_A^{idb} \leftarrow \mathsf{rpk}_{0}^{dh}$, $\text{DH}_B^{idb} \leftarrow \bot$\\
  -- holds no ratchet key state, since it never learns $K_S$
\item IdP sets its initial state:\\
  -- initializes $\text{RK} \leftarrow K_S$, $\text{CK} \leftarrow \bot$\\
  -- generates DH key pair $(\mathsf{isk}^{dh}, \mathsf{ipk}^{dh})$\\
  -- sets current DH key state $\text{DH}_A^{idp} \leftarrow \bot$, $\text{DH}_B^{idp} \leftarrow \mathsf{ipk}^{dh}$
\end{itemize}

\textbf{Broker:}\\
\textbf{On receiving} token request from the RP and the access token is expired:
\begin{enumerate}[leftmargin=*]
  \item constructs message header $\mathsf{hdr} = \text{DH}_A^{idb}$
  \item generates token request message: $\mathsf{msg} = (\mathsf{hdr}, seq, \mathsf{refresh\_token})$
\end{enumerate}

\textbf{On receiving} $(\mathsf{hdr}, seq, c)$ from the IdP:
\begin{enumerate}[leftmargin=*]
  \item if $\text{DH}^{idb}_B \neq \mathsf{hdr}$:
  \begin{itemize}
    \item updates $	\text{DH}_B^{idb} \leftarrow \mathsf{hdr}$
    \item advances local key: $seq \leftarrow seq+1$, $ \text{DH}_A^{idb} \leftarrow \mathsf{rpk}_{seq}^{dh}$
  \end{itemize}
\end{enumerate}

\textbf{IdP:}\\
\textbf{On receiving} $\mathsf{msg} = (\mathsf{hdr}, seq, \mathsf{refresh\_token})$ from the broker
\begin{enumerate}[leftmargin=*]
  \item if $	\text{DH}_A^{idp} \neq \mathsf{hdr}$:
  \begin{itemize}
    \item updates $	\text{DH}_A^{idp} \leftarrow \mathsf{hdr}$
    \item generates new $(\mathsf{isk}^{dh}, \mathsf{ipk}^{dh})$ and updates $\text{DH}_B^{idp} \leftarrow \mathsf{ipk}^{dh}$
    \item updates $(\text{RK}, \text{CK}) \leftarrow \mathsf{KDF\_RK}(\text{RK}, \mathsf{DH}(\mathsf{isk}^{dh}, 	\text{DH}_A^{idp}))$
  \end{itemize}
  \item derives $\text{CK}, \text{MK} \leftarrow \mathsf{KDF\_CK}(\text{CK})$
  \item encrypts $c = \mathsf{ENC}_{\text{MK}}(m)$
  \item construct message header $\mathsf{hdr} = \text{DH}_B^{idp}$
  \item sends $(\mathsf{hdr}, seq, c)$ to RP via Broker
\end{enumerate}

\textbf{RP:}\\
\textbf{On receiving} $(hdr, seq, c)$ from the broker:
\begin{enumerate}[leftmargin=*]
\item if $\text{DH}_B^{rp} \neq \mathsf{hdr}$:
  \begin{itemize}
  \item updates $\text{DH}_B^{rp} \leftarrow \mathsf{hdr}$
  \item gets private key $\mathsf{rsk}_{dh} \leftarrow \mathsf{rsk}_{seq}^{dh}$
  \item updates $(\text{RK}, \text{CK}) \leftarrow \mathsf{KDF\_RK}(\text{RK}, \mathsf{DH}(\mathsf{rsk}_{dh}, 	\text{DH}_B^{rp}))$
  \end{itemize}
  \item derives new message key $\text{CK}, \text{MK} \leftarrow \mathsf{KDF\_CK}(\text{CK})$
  \item decrypts $m = \mathsf{DEC}_{\text{MK}}(c)$
\end{enumerate}
\end{footnotesize}
\end{mdframed}
\caption{Token Refresh with Key Rotation Mechanism}
  \label{fig:Token Refresh Protocol with Key Rotation Mechanism}
\end{figure}

\section{Trusted Computation Agent}\label{app:TCA}

The Trusted Computation Agent (TCA) is the sole client-side trusted compute unit of \name{}, instantiated inside the user's browser as a cross-origin \texttt{<iframe>}. It executes user-side cryptographic operations on the user's behalf, removing the cryptographic burden from the RP client and avoiding any user-side deployment such as a browser extension. As the sole \emph{client-side trust anchor} of \name{}, the TCA must satisfy two security goals.
\begin{itemize}[leftmargin=*]
    \item \emph{Trusted execution of cryptographic operations.} The TCA must faithfully execute the prescribed user-side cryptographic operations, in particular $\mathsf{AC.VerifyCred}$ for consent-time RP verification and the OPRF subprotocol for PIN-derived key generation, so that the $D$ rendered on the consent page is exactly the value the IdP signed into $acid$ and the management key $K_m$ is derived from a genuine OPRF evaluation rather than a malicious substitute.
    \item \emph{Confidentiality of process data.} Intermediate values produced or held during TCA execution, including the PIN, the management key $K_m$, the OPRF blinding factor, and the plaintext $D$, must not become observable to the IdP back end, the broker, or any third-party script in the user's browser.
\end{itemize}

\noindent\textbf{Implementation Details.}
The TCA is loaded as a cross-origin \texttt{<iframe>} embedded in the relevant IdP-hosted pages and performs two functions tied to specific \name{} phases.
\begin{itemize}[leftmargin=*]
    \item \emph{Consent-time RP verification} (Section~\ref{sec:User Authentication and Consent}). The broker iframe forwards $(cid, acid, \pi_D, D)$ to the TCA via \texttt{window.postMessage}. The TCA then parses $\pi_D = t_1$ and evaluates $\mathsf{AC.VerifyCred}(\mathsf{ivk}_c, D, acid, \pi_D)$, which derandomizes $acid$ with $t_1$ before checking the signature, and, on success, renders $D$ inside its own iframe surface for user confirmation.
    \item \emph{Private authorization management} (Section~\ref{sec:Private Authorization Management}). The broker iframe delivers $(acid, D_{\mathsf{Enc}})$, the TCA prompts the user for a short PIN, runs the OPRF subprotocol with the IdP to derive $K_m = \mathsf{KDF}(\mathsf{Unblind}(r', \mathsf{BlindEval}(\mathsf{k}_{oprf}, \mathsf{Blind}(\mathsf{PIN}))))$, decrypts $D \leftarrow \mathsf{Dec}_{K_m}(D_{\mathsf{Enc}})$, renders the plaintext $D$ for review or revocation, and emits only $(acid, \text{decision})$ to the broker.
\end{itemize}

\noindent\textbf{TCA Security.}
In practice, the TCA's security is achieved not through bespoke cryptographic constructions but through a small set of concrete frontend deployment strategies that modern browsers provide natively. We combine three such standard web-platform primitives to jointly enforce the two goals above.
\begin{itemize}[leftmargin=*]
    \item \emph{Browser same-origin policy.} Because the TCA, the broker iframe, and the IdP page each live in distinct origins, the browser same-origin policy prevents any IdP-hosted script from reading TCA-internal state. The only cross-frame channel is \texttt{window.postMessage}. Each sender pins the recipient with an explicit \texttt{targetOrigin}, and the receiving frame accepts a message only after checking \texttt{event.origin} and \texttt{event.source} against the expected counterpart, validating the message structure, and matching a fresh per-session nonce bound to the session identifier, so that cross-origin injection and replayed messages are rejected.
    \item \emph{Runtime integrity via Subresource Integrity.} The TCA document loads its external script bundle through an SRI-pinned \texttt{<script>} element, so the browser refuses to execute any bundle whose bytes deviate from the pinned hash, binding runtime execution to the exact release named by the versioned bundle URL. Users can also recompute the bundle hash on demand to independently verify its integrity.
    \item \emph{Auditable and upgradeable provenance.} The TCA bundle is served from a publicly auditable host independent of any IdP, broker, or RP, and the reference hash of every release is recorded in an append-only binary transparency log analogous to Certificate Transparency. This supports controlled version upgrades while making silent divergence between users, as well as retroactive replacement of a past release, globally detectable.
\end{itemize}

\noindent\textbf{Hardening against a Malicious IdP.}
The mechanisms above suffice under our honest-but-curious model. Stronger guarantees against an actively malicious IdP can be obtained by redistributing the canonical TCA bundle through out-of-band channels (e.g., a browser extension, a signed WebAssembly module, or RP-operated transparency-log mirrors), so that bundle delivery no longer depends on the IdP-served \texttt{<iframe>} tag. A strictly malicious IdP combining a backdoored bundle, a forged TCA host, or browser zero-days falls outside this model and is left to future work, aligning with how prior privacy-preserving SSO systems such as ARPSSO~\cite{he2024arpsso} treat equivalent client-side components.

\section{Additional Discussion} \label{app:additional discussion}
\re{
\noindent\textbf{Scope of RP Anonymity.}
RP anonymity in \name{} is defined relative to the IdP rather than to the broker. The protected property is the \emph{user--RP authorization relationship}, namely which user has authorized which RP, and it is the IdP from which this relationship stays hidden. The broker needs $D$ and $cid$ in the clear to retrieve the RP configuration, route the request, and bind the IdP response back to the RP session, and the $\mathsf{BrokerView}$ outputs in Figure~\ref{fig:ideal-orionlink} grant this metadata explicitly. Knowing the RP therefore only lets the broker associate it with the per-RP pseudonym $uid_{rp}$, which carries no user attributes and stays unlinkable across RPs, so Definitions~\ref{def:untraceability} and~\ref{def:unlinkability} are unaffected.

\noindent\textbf{Discussion of Forward Secrecy.}
\ret{In our X3DH variant the RP contributes session-level secrets, the ephemeral
$\mathsf{rsk}_E$ and the randomizer $t_3$, which it erases after each session,
while the IdP contributes its long-lived identity key
$\mathsf{isk}_{IK}$ and a signed prekey $\mathsf{isk}_E$ refreshed once per
epoch. Losing both IdP secrets exposes one prekey epoch rather than all past
sessions, and shortening the epoch interval mitigates this exposure.}

\ret{\name{} can be upgraded to session-level forward secrecy at an acceptable
additional cost. In step~\circnumtwo{8} of
Figure~\ref{fig:overview of OrionLink} the IdP samples a one-time key pair
$(\mathsf{isk}_O,\mathsf{ipk}_O)$, derives
$K_S'=\mathsf{KDF}(K_S\|\mathsf{DH}(\mathsf{isk}_O,\mathsf{rpk}_E))$, erases
$\mathsf{isk}_O$, and returns $\mathsf{ipk}_O$ alongside $\mathsf{info}_{Enc}$.
On receiving the response, the RP rebuilds $K_S'$ from $\mathsf{rsk}_E$, and
both sides use it in place of $K_S$. Refreshed tokens already ratchet in fresh
secrets and are unaffected.}

\noindent\textbf{Hardening Private Authorization Management.}
The private authorization management of Section~\ref{sec:Private Authorization Management} is an optional feature that lets a user review and revoke authorizations from any device with a memorized PIN. Because the key is derived through an OPRF in the manner of the aPAKE design OPAQUE~\cite{jarecki2018opaque}, no party that lacks the OPRF key $\mathsf{k}_{oprf}$ can mount an offline dictionary attack on the PIN, which covers the broker, the RPs, a network attacker, and anyone who obtains the stored $D_{\mathsf{Enc}}$.
One limitation remains, since the IdP both holds $\mathsf{k}_{oprf}$ and stores $D_{\mathsf{Enc}}$ and can therefore enumerate a short PIN offline in principle. A deployment can remove this capability with a $(2,2)$-distributed OPRF evaluated jointly by the IdP and the broker~\cite{jarecki2014round}, which turns every guess into a rate-limited online interaction under our non-collusion assumption, or with a key bound to user-side hardware such as a WebAuthn passkey, at some cost in convenience.

\noindent\textbf{Post-Quantum Migration.}
\name{} composes standard primitives, so a post-quantum upgrade can proceed primitive by primitive. The X3DH core of AAKA can be replaced by PQXDH~\cite{kret2023pqxdh}, which Signal already deploys at scale, and the PS-based credential by a lattice-based anonymous credential scheme~\cite{bootle2023framework} exposing the same issuance and presentation interface. The DH-based OPRF that derives $uid_{rp}$ and $K_m$ has lattice-based candidates~\cite{albrecht2021round}. A full post-quantum instantiation and its evaluation remain future work.
}
\section{Preliminary}
In this section, we describe the security properties required of the
building blocks used in our construction.
Let $\lambda$ denote the security parameter, and let $\mathcal{A}$ be a probabilistic polynomial-time (PPT) adversary.
\subsection{Security Properties of OPRFs}\label{sec:Security Properties of OPRFs}
Under the Decisional Diffie--Hellman (DDH) problem \cite{boneh1998decision} in the random oracle model, the OPRF used in our construction satisfies correctness,
consistency, and pseudorandomness.
\begin{enumerate}[leftmargin=*]
\item \textit{Correctness.}
For any $\lambda$, any secret key $k$, and any input $x$,
an honest execution of the $\mathsf{Blind}(\cdot)$, $\mathsf{BlindEval}(\cdot)$, and
$\mathsf{Unblind}(\cdot)$ algorithms outputs $\mathsf{OPRF}_k(x)$,
except with $\negl(\lambda)$.
\item \textit{Consistency.}
For any $\lambda$, for any $\mathcal{A}$, there exists a negligible function
$\negl(\cdot)$ such that
\[
\Pr[\mathsf{Exp}^{\mathsf{cons}}_{\mathsf{OPRF}}(\mathcal{A},\lambda)=1]
\ge 1-\negl(\lambda),
\]
where the experiment $\mathsf{Exp}^{\mathsf{cons}}_{\mathsf{OPRF}}(\mathcal{A},\lambda)$
is defined as follows:
\begin{itemize}[leftmargin=*]
\item The challenger samples a secret key $k \xleftarrow{\$} \mathbb{Z}_p$
and an input $x$.
\item The challenger runs two independent honest executions of the OPRF
protocol on input $x$, obtaining outputs $y_1$ and $y_2$.
\item The experiment outputs $1$ iff $y_1 = y_2$.
\end{itemize}
\item \textit{Pseudorandomness.}
For any $\lambda$, for any $\mathcal{A}$, there exists a negligible function $\negl(\cdot)$ such that
\[
\Big| \Pr[\mathsf{Exp}^{\mathsf{pseud}}_{\mathsf{OPRF}}(\mathcal{A},\lambda)=1] - \tfrac12 \Big| \le \negl(\lambda)
\]
define the experiment $\mathsf{Exp}^{\mathsf{pseud}}_{\mathsf{OPRF}}(\mathcal{A},\lambda)$ as follows:
\begin{itemize}[leftmargin=*]
\item The challenger samples a secret key
$k \xleftarrow{\$} \mathbb{Z}_p$
and keeps it secret.

\item $\mathcal{A}$ is given adaptive oracle access to $\mathsf{OPRF}_k(\cdot)$, implemented via the honest $\mathsf{BlindEval}(\cdot)$ and $\mathsf{Unblind}(\cdot)$.

\item Eventually, $\mathcal{A}$ outputs a challenge input $x'$ that has not been queried before.

\item The challenger samples a random bit $b \xleftarrow{\$} \{0,1\}$.
If $b=0$, it sets $y' := \mathsf{OPRF}_k(x')$;
if $b=1$, it sets $y' \xleftarrow{\$} \mathbb{G}$.
The challenger returns $y'$ to $\mathcal{A}$.
\item $\mathcal{A}$ outputs a guess $b'$.
\item The experiment outputs $1$ iff $b'=b$.
\end{itemize}
\end{enumerate}

\subsection{Security Properties of ACs} \label{sec:Security Properties of Anonymous Credential}
Under the XDH assumption, implying the hardness of the Computational
co-Diffie--Hellman (co-CDH) problem in $\mathbb{G}_1$ and $\mathbb{G}_2$ and the
Decisional Diffie--Hellman (DDH) problem in $\mathbb{G}_1$ \cite{pointcheval2016short}, where all zero-knowledge proofs are instantiated using the Fiat--Shamir transform of underlying $\Sigma$-protocols. The anonymous credential protocol that we used satisfies correctness, unforgeability, blindness, unlinkability/zero-knowledge. 
\begin{enumerate}[leftmargin=*]
\item \textit{Correctness.}
For any $\lambda$, and any attribute vector
$A=(A_h,A_p)$ such that $\phi(A_h)=1$, where $\phi$ is an efficiently computable predicate on hidden attributes, it holds that
\[
\Pr\left[
\begin{aligned}
&(pp)\leftarrow \mathsf{AC.Setup}(1^\lambda);\\
&(sk,vk) \leftarrow \mathsf{AC.KeyGen}(pp,n);\\
&(A,\sigma)\leftarrow \mathsf{AC.IssueCred}(pp,sk,A_h,A_p);\\
&(A'_h,A'_p,\pi,\phi)\leftarrow \mathsf{AC.PresentCred}(vk,\sigma,A):\\
&\mathsf{AC.VerifyCred}(vk,A'_p,\pi,\phi)=1
\end{aligned}
\right]=1.
\]

\item \textit{Unforgeability.}
For any $\lambda$, for any $\mathcal{A}$, there exists a negligible function $\negl(\cdot)$ such that
\[
\Pr\!\left[\mathsf{Exp}^{\mathsf{uf}}_{\mathsf{AC}}(\mathcal{A},\lambda)=1\right]
\le \negl(\lambda).
\]
where the experiment
$\mathsf{Exp}^{\mathsf{uf}}_{\mathsf{AC}}(\mathcal{A},\lambda)$ is defined as follows:

\begin{itemize}[leftmargin=*]
\item The challenger runs
$$pp \leftarrow \mathsf{AC.Setup}(1^\lambda),
(sk,vk) \leftarrow \mathsf{AC.KeyGen}(pp,n)$$
and gives $(pp,vk)$ to $\mathcal{A}$.

\item $\mathcal{A}$ is given oracle access to $\mathsf{AC.IssueCred}(\cdot)$ and $\mathsf{AC.PresentCred}(\cdot)$. 
The challenger maintains a set $\mathcal{Q}$ containing all pairs
$(A,\pi)$ that arise from honest executions of these oracles.
\item Eventually, $\mathcal{A}$ outputs a tuple
$(\tilde{A},\tilde{\sigma},\tilde{\pi},\phi)$.
\item The experiment outputs $1$ if
\[
\mathsf{AC.VerifyCred}(\tilde{A},\tilde{\sigma},\tilde{\pi},\phi)=1
\quad\text{and}\quad
(\tilde{A},\tilde{\pi}) \notin \mathcal{Q};
\]
otherwise, the experiment outputs $0$.
\end{itemize}
\item \textit{Blindness.}
For any $\mathcal{A}$ controlling the issuer, there exists a negligible
function $\negl(\cdot)$ such that
\[
\Big|\Pr\!\left[\mathsf{Exp}^{\mathsf{blind}}_{\mathsf{AC}}(\mathcal{A},\lambda)=1\right]-\tfrac12\Big|
\le \negl(\lambda).
\]
where the experiment $\mathsf{Exp}^{\mathsf{blind}}_{\mathsf{AC}}(\mathcal{A},\lambda)$  proceeds as follows:
\begin{itemize}[leftmargin=*]
  \item $\mathcal{A}$ samples two attribute vectors
  $A^{(0)}=(A_h^{(0)},A_p)$ and
  $A^{(1)}=(A_h^{(1)},A_p)$ such that
  $\phi(A_h^{(0)})=\phi(A_h^{(1)})=1$.
  \item A challenger samples $b\leftarrow\{0,1\}$ and runs
  \[
  (A,\sigma)\leftarrow \mathsf{AC.IssueCred}(pp,sk,A^{(b)}),
  \]
  interacting with $\mathcal{A}$ as the issuer.
  \item $\mathcal{A}$ outputs a guess $b'$.
   \item The experiment outputs $1$ iff $b'=b$.
\end{itemize}

\item \textit{Unlinkability/Zero-Knowledge.}
For any $\mathcal{A}$, including one that generates the key pair $(sk,vk)$ itself and therefore controls both credential issuance and verification, there exists a negligible function $\negl(\cdot)$ such that
\[
\Big|\Pr\!\left[\mathsf{Exp}^{\mathsf{unlink}}_{\mathsf{AC}}(\mathcal{A},\lambda)=1\right]-\tfrac12\Big|
\le \negl(\lambda).
\]
where the experiment $\mathsf{Exp}^{\mathsf{unlink}}_{\mathsf{AC}}(\mathcal{A},\lambda)$ is defined as follows:
\begin{itemize}[leftmargin=*]
  \item $\mathcal{A}$ generates the key pair $(sk,vk)$ and outputs $vk$ together with a predicate $\phi$.
  \item The challenger samples two attribute vectors
  $A^{(0)}=(A_h^{(0)},A_p)$ and
  $A^{(1)}=(A_h^{(1)},A_p)$ such that
  $\phi(A_h^{(0)})=\phi(A_h^{(1)})=1$.
  \item The challenger obtains two credentials
  $(A^{(0)},\sigma^{(0)})$ and
  $(A^{(1)},\sigma^{(1)}) $ by executing $\mathsf{AC.IssueCred}$ with $\mathcal{A}$ as the issuer.

  \item A bit $b\leftarrow\{0,1\}$ is sampled.
  \item If $b=0$, the challenger runs
  \[
  \begin{aligned}
  &T_1 \leftarrow \mathsf{AC.PresentCred}(vk,\sigma^{(0)},A^{(0)}),\\
  &T_2 \leftarrow \mathsf{AC.PresentCred}(vk,\sigma^{(0)},A^{(0)}).
  \end{aligned}
  \]
  \item If $b=1$, the challenger runs
  \[
  \begin{aligned}
  &T_1 \leftarrow \mathsf{AC.PresentCred}(vk,\sigma^{(0)},A^{(0)}),\\
  &T_2 \leftarrow \mathsf{AC.PresentCred}(vk,\sigma^{(1)},A^{(1)}).
  \end{aligned}
  \]
  \item The challenger gives $(T_1,T_2)$ to $\mathcal{A}$, which outputs a guess $b'$.
  \item The experiment outputs $1$ iff $b'=b$.
\end{itemize}

\end{enumerate}

\section{Security Proof of \textit{OrionLink}} \label{app:security proof}

Recall Theorem \ref{thm:security}. We now prove that the protocol in Figure~\ref{fig:The OrionLink Protocol} securely realizes $\mathcal{F}_{\mathsf{OrionLink}}$. Specifically, we show that for any real-world  adversary $\mathcal{A}$, we can construct an ideal world simulator $\mathcal{S}$, such  that for all environments $\mathcal{Z}$, the ideal execution with $\mathcal{S}$ is indistinguishable from the real execution with $\mathcal{A}$. Throughout, parties communicate over authenticated secure channels, $\mathcal{F}_{\mathsf{OrionLink}}$ delivers its outputs immediately.

We consider four distinct adversary models based on the trust models:
\begin{enumerate}
    \item \textbf{Malicious RP ($\mathcal{A}_{RP}$):} Controls the RP to violate authentication correctness or user identity unlinkability.
    \item \textbf{Semi-Honest IdP ($\mathcal{A}_{IdP}$):} Follows the protocol but attempts to trace user behavior (i.e., link users with the RPs they access).
    \item \textbf{Semi-Honest Broker ($\mathcal{A}_{\mathsf{IdB}}$):} Follows the protocol but attempts to violate broker obliviousness (i.e., learn user sensitive data).
    \re{\item \textbf{Joint Corruption ($\mathcal{A}_{\mathsf{Joint}}$):} The semi-honest broker pools its view with an arbitrary set of malicious colluding RPs, attempting to link user pseudonyms across RPs.}
\end{enumerate}

For the sake of completeness of the threat model, we consider users to be potentially malicious. However, since OrionLink inherits the user management and authentication mechanisms of standard SSO systems, it also inherits their user authentication correctness guarantees.

\subsection{UC Simulation Proof}
We prove Theorem~\ref{thm:security} through the four corruption cases listed above.
\subsubsection*{Case 1: Malicious RP ($\mathcal{A}_{\mathsf{RP}}$)}
Ensure that the RP cannot break authentication correctness or link a user’s identity across multiple RPs.

\textbf{Simulator} $\mathcal{S}_{\mathsf{RP}}$
runs $\mathcal{A}_{\mathsf{RP}}$ internally, and simulates the honest IdP, broker, and user.

\noindent\textbf{Setup.} 
$\mathcal{S}_{\mathsf{RP}}$ initializes the public parameters $\mathsf{pp}$ and samples the IdP's long-term keys exactly as in \textit{Initialization phase} in Figure \ref{fig:The OrionLink Protocol}.

\noindent\textbf{Registration.} Upon receives  $(D,\pi_1)$ from $\mathcal{A}_{\mathsf{RP}}$:
\begin{enumerate}[leftmargin=*]
\item  $\mathcal{S}_{RP}$ sends $(\mathsf{RegisterRP}, RP, D)$ to $\mathcal{F}_{\mathsf{OrionLink}}$, receives ($cid$).
\item $\mathcal{S}_{RP}$ signs as IdP with $cid$:
$\sigma \leftarrow \mathsf{PS.Sign}(\mathsf{isk}_c,D)$ and
$\hat{\bar{\sigma}}\leftarrow \mathsf{BlindSign}(\mathsf{isk}_c,D,\pi_1)$,
then returns $(cid,\sigma,\hat{\bar{\sigma}})$ to $\mathcal{A}_{\mathsf{RP}}$, and locally records $(cid,\mathcal{DC})$ in the broker state.
\end{enumerate}

\noindent\textbf{Login Request.} $\mathcal{S}_{\mathsf{RP}}$ receives $(cid,acid,\pi_D,M)$ from $\mathcal{A}_{\mathsf{RP}}$, where $acid$ is the randomized-signature session pseudonym, $\pi_D=t_1$ is the disclosed derandomizer, and $M=(\pi_2,\bar{\sigma}',\mathsf{rpk}'_{IK},\mathsf{rpk}_E,\mathsf{msg}_{init})$ carries the anonymous authentication proof and the key-agreement payload.
\begin{enumerate}[leftmargin=*]
  \item $\mathcal{S}_{RP}$ runs the honest IdP's verification of the presentation, evaluating $\mathsf{VerifyProof}(\pi_2,\sigma',\bar{\sigma}',\mathsf{rpk}'_{IK},(\mathsf{isk}_c,\mathsf{ivk}_c))$ with $\sigma'=acid$ exactly as in the real protocol, and records the outcome.
  \item $\mathcal{S}_{RP}$ sends $(\mathsf{InitLogin},RP,sid,cid,acid)$ to $\mathcal{F}_{\mathsf{OrionLink}}$. Because the RP is corrupted, $\mathcal{F}_{\mathsf{OrionLink}}$ adopts the supplied $acid$ and issues the query $(\mathsf{Judge},sid,cid,acid)$, which $\mathcal{S}_{RP}$ answers with the recorded verification outcome, so the ideal session proceeds exactly when the real IdP would accept the presentation. $\mathcal{S}_{RP}$ then relays the remaining messages to the intended receivers.
\end{enumerate}
\noindent\textbf{User Authentication \& Consent.}
\begin{enumerate}[leftmargin=*]
\item $\mathcal{S}_{RP}$ forwards $(D,\pi_D)$ to the simulated TCA and checks $\mathsf{AC.VerifyCred}(\mathsf{ivk}_c,D,acid,\pi_D)$, aborting if the check fails.
\item  $\mathcal{S}_{RP}$ sends $(\mathsf{ShowRP},RP,sid,cid,D)$ to $\mathcal{F}_{\mathsf{OrionLink}}$.
\item When the user returns $decision\in\{\mathsf{approve},\mathsf{deny}\}$, $\mathcal{S}_{RP}$ sends $(\mathsf{Authorize},User,sid,uid,decision)$ to $\mathcal{F}_{\mathsf{OrionLink}}$.
\end{enumerate}
\textbf{Login Response}. If $\mathcal{F}_{\mathsf{OrionLink}}$ outputs
$(\mathsf{AuthOK}, sid, uid_{rp}, Info)$, $\mathcal{S}_{\mathsf{RP}}$ proceeds as follows:
\begin{enumerate}[leftmargin=5mm]
\item $\mathcal{S}_{\mathsf{RP}}$ maps the ideal identifier $uid_{rp}$ to
$\mathbb{Z}_p$ by applying a hash-to-field function
$H_{\mathbb{Z}_p}$, extracts the group element $\ret{\sigma_1'}$ from the RP pseudonymous
$acid=\ret{(\sigma_1',\sigma_2')}$, and computes $auid:= \ret{\sigma_1'^{\,uid_{rp}}}$,
which is sent to $\mathcal{A}_{\mathsf{RP}}$ so that unblinding yields $uid_{rp}$.

\item
Based on the previously received $(acid,M)$
and the honest Id’s keys materials,
$\mathcal{S}_{RP}$ derives $K_S$ as the honset IdP does, and encrypts $\mathsf{info}_{Enc}=Enc_{K_S}(\mathsf{info})$.
\item Sends $(auid,\mathsf{info}_{Enc})$ to $\mathcal{A}_{\mathsf{RP}}$.
\end{enumerate}

Upon receiving $uid_{rp}$ from $\mathcal{A}_{\mathsf{RP}}$, the simulated broker records $(cid, uid_{rp})$.

Otherwise, $\mathcal{S}_{\mathsf{RP}}$ aborts the session and returns a
rejection to the RP, simulating a failed credential verification.

Now we argue that the ideal execution with $\mathcal{S}_{\mathsf{RP}}$ is indistinguishable from the real execution with $\mathcal{A}_{\mathsf{RP}}$.

\textbf{Hybrid} $\textbf{H}_0$
This is the real execution of $\Pi_{\mathsf{OrionLink}}$ with $\mathcal{A}_{\mathsf{RP}}$ controlling the RP.

\textbf{Hybrid} $\textbf{H}_1$ is the same as $H_0$, except that in the registration phase $\mathcal{S}_{\mathsf{RP}}$  computes
$(\sigma,\hat{\bar{\sigma}})$ using the same signing and blind-signing
algorithms as the real IdP and returns
$(cid,\sigma,\hat{\bar{\sigma}})$ to $\mathcal{A}_{\mathsf{RP}}$.
$\mathcal{S}_{\mathsf{RP}}$ samples all keys and randomness
according to the same distributions as the honest IdP, $H_0$ and $H_1$ are perfectly indistinguishable.

\textbf{Hybrid} $\textbf{H}_2$ is the same as $H_1$, except that 
during the login request and user-consent phases, $\mathcal{S}_{\mathsf{RP}}$ simply forwards all messages between the parties and the ideal functionality $\mathcal{F}_{\mathsf{OrionLink}}$. 

Hence, $H_1$ and $H_2$ are perfectly indistinguishable.

\textbf{Hybrid} $\textbf{H}_3$  is the same as $H_2$, except that the acceptance condition of a login session is made explicit
and aligned between the real world and the ideal functionality. In the real world, the IdP proceeds to generate a login response
if and only if
\[
1\leftarrow \mathsf{VerifyProof}(\pi_2,\sigma',\bar{\sigma}',
\mathsf{rpk}'_{IK},(\mathsf{isk}_c,\mathsf{ivk}_c)),
\]
otherwise the session aborts.
In $\mathcal{F}_{\mathsf{OrionLink}}$, the session is finalized
if and only if (i) the user has authorized the session and (ii) the finalize message carries the same $acid$ recorded in $\mathsf{Sess}[sid]$, a binding that $\mathcal{F}_{\mathsf{OrionLink}}$ created itself for an honest RP and, for the corrupted RP, created only after checking that $\mathsf{Owner}[cid]$ is corrupted and that the adversary replied $\mathsf{accept}$ to the $\mathsf{Judge}$ query.

Consider the two possible mismatch events.

\re{
(1) \textit{Real accepts but Ideal rejects.}
Suppose the IdP-side verification in the real world outputs $1$ while the ideal session is not finalized.
If the user did not authorize, the real world also produces no response.
An abort in the $\mathsf{Judge}$ step is impossible, because $\mathcal{S}_{\mathsf{RP}}$ answers that query by running exactly the real verification.
The remaining case is an abort because $\mathsf{Owner}[cid]$ is honest, meaning that $\mathcal{A}_{\mathsf{RP}}$ produced a presentation $(\sigma',\bar{\sigma}',\pi_2)$ that $\mathsf{VerifyProof}$ accepts for a $cid$ registered by an honest RP. We then construct $\mathcal{B}$ that wins $\mathsf{Exp}^{\mathsf{uf}}_{\mathsf{AC}}$ by outputting this tuple, which is accepted without arising from any credential issued to the adversary. By the unforgeability of $\mathsf{AC}$, this event occurs only with negligible probability.

(2) \textit{Ideal accepts but Real rejects.}
For the corrupted RP this event is excluded by construction, because $\mathcal{F}_{\mathsf{OrionLink}}$ finalizes only sessions whose presentation $\mathcal{S}_{\mathsf{RP}}$ accepted by running the real verification.
For an honest RP the event means that a presentation honestly derived from a valid issuer-generated credential is rejected by $\mathsf{VerifyCred}$, which contradicts the correctness of $\mathsf{AC}$ and hence occurs only with negligible probability.
}
Since both mismatch events are ruled out by the unforgeability and the correctness of $\mathsf{AC}$, $H_2$ and $H_3$ are computationally indistinguishable.

\textbf{Hybrid} $\textbf{H}_4$ is the same as $H_3$, except for the RP-specific pseudonym $auid$ returned to the RP. 
From the view of $\mathcal{A}_{\mathsf{RP}}$, the real user identity $uid$ is never
revealed. The RP only obtains $uid_{rp}$, which is pseudorandom except for consistency across sessions with the same User--RP pair.
$\mathcal{S}_{\mathsf{RP}}$ receives a uniformly random
$uid_{rp}$ from $\mathcal{F}_{\mathsf{OrionLink}}$. If two sessions correspond to the same
User--RP pair, $\mathcal{S}_{\mathsf{RP}}$ outputs consistent values of $auid$,
otherwise $auid$ is generated pseudorandomly according to the OPRF.

Hence, $H_3$ and $H_4$ are computationally indistinguishable under the pseudorandomness of the OPRF.

\textbf{Hybrid} $\textbf{H}_5$ is the same as $H_4$, except for the derivation of the session key $K_S$. $\mathcal{S}_{\mathsf{RP}}$ derives the session key $K_S$ by honestly executing
$\mathsf{X3DH}$ exactly as an honest IdP would do in the real world.

Hence, $H_4$ and $H_5$ are perfectly indistinguishable.

Since $H_0 \approx_c H_5$, the real execution with $\mathcal{A}_{\mathsf{RP}}$ is computationally indistinguishable from the ideal execution with $\mathcal{S}_{\mathsf{RP}}$ and $\mathcal{F}_{\mathsf{OrionLink}}$.

\subsubsection*{Case 2: Semi-Honest IdP ($\mathcal{A}_{\mathsf{IdP}}$)}
Ensure the IdP cannot identify which RP the user is accessing or link multiple sessions.

\textbf{Simulator $\mathcal{S}_{\mathsf{IdP}}$.}
Since the IdP is semi-honest, $\mathcal{S}_{\mathsf{IdP}}$ simulates the honest RP, broker and user, and relays all messages between $\mathcal{A}_{\mathsf{IdP}}$ and
$\mathcal{F}_{\mathsf{OrionLink}}$.

\textbf{Hybrid} $\textbf{H}_0$
This is the real execution of $\Pi_{\mathsf{OrionLink}}$ with $\mathcal{A}_{\mathsf{IdP}}$, This hybrid corresponds to the real execution of
$\Pi_{\mathsf{OrionLink}}$ with $\mathcal{A}_{\mathsf{IdP}}$.
In a login session, the IdP receives a login request forwarded by the broker,
consisting of
\[
(acid, M),
\quad
M=(\bar{\sigma}', \pi_2, \mathsf{rpk}'_{IK}, \mathsf{rpk}_E, \mathsf{msg}_{init}),
\]
which are the RP's anonymous credentials.

\textbf{Hybrid} $\textbf{H}_1$ is the same as $H_0$, except that the credential materials $(acid, \bar{\sigma}', \pi_2)$ are generated by the simulator $\mathsf{AC.PresentCred(\cdot)}$ and $\mathsf{GenProof}(\cdot)$, without using any RP-specific secret or identity.
The simulator produces a tuple $(acid^{\mathsf{sim}}, \bar{\sigma}'^{\mathsf{sim}}, \pi_2^{\mathsf{sim}}, {rpk}'^{\mathsf{sim}}_{IK})$, 
such that the IdP's verification procedure accepts 
and whose distribution is computationally indistinguishable from that of an honestly generated anonymous credential presentation.
$\mathcal{A}_{\mathsf{IdP}}$ cannot distinguish $H_0$ and $H_1$.
Otherwise, we can construct a $\mathcal{B}$ to win the unlinkability experiment
$\mathsf{Exp}^{\mathsf{unlink}}_{\mathsf{AC}}$ of the anonymous credential scheme
with non-negligible advantage.

Hence, $H_0$ and $H_1$ are computationally indistinguishable under the unlinkability of $\mathsf{AC}$.

\textbf{Hybrid} $\textbf{H}_2$ is the same as $H_1$, except that the remaining components
$(\mathsf{rpk}_E, \mathsf{msg}_{init})$ are generated according to their protocol
specification.
Since $\mathsf{rpk}_E$ is an ephemeral, one-time public key sampled uniformly at
random, and $\mathsf{msg}_{init}$ is an encryption of values derived solely from
such ephemeral randomness, these components are independent of the RP's
identity. $H_1$ and $H_2$ are perfectly indistinguishable.

Since $H_0 \approx_c H_2$, the real execution with $\mathcal{A}_{\mathsf{IdP}}$ is computationally indistinguishable from the ideal execution with $\mathcal{S}_{\mathsf{IdP}}$ and $\mathcal{F}_{\mathsf{OrionLink}}$.

\subsubsection*{Case 3: Semi-Honest Broker ($\mathcal{A}_{\mathsf{IdB}}$)}
Ensure that the view of $\mathcal{A}_{\mathsf{IdB}}$ can be simulated from the leakage $\mathcal{L}_{\mathsf{IdB}}$ alone, as required by Definition~\ref{def:broker-obliviousness}. In particular, the broker learns neither the user's real identity, nor the contents of $info$, nor the session key $K_S$.

\textbf{Simulator $\mathcal{S}_{\mathsf{IdB}}$.}
Since the broker is semi-honest, $\mathcal{S}_{\mathsf{IdB}}$ simulates the honest IdP, RP, and user, and relays all messages between $\mathcal{A}_{\mathsf{IdB}}$ and
$\mathcal{F}_{\mathsf{OrionLink}}$.

\textbf{Hybrid} $\textbf{H}_0$ The real world protocol execution.

\textbf{Hybrid} $\textbf{H}_1$ is the same as $H_0$, except that every message in the broker's view other than the ciphertext $\mathsf{info}_{Enc}$ is produced by $\mathcal{S}_{\mathsf{IdB}}$ from the leakage $\mathcal{L}_{\mathsf{IdB}}$ and the randomness of the simulated honest parties.
Upon $(\mathsf{BrokerView}, cid, D)$, $\mathcal{S}_{\mathsf{IdB}}$ runs the registration phase for $D$ on behalf of the simulated IdP, whose keys it holds.
Upon $(\mathsf{BrokerView}, sid, cid, acid)$, it runs $\mathsf{AC.PresentCred}(\cdot)$ and $\mathsf{GenProof}(\cdot)$ on behalf of the simulated RP with fresh randomizers $(k_1, t_1)$ and ephemeral keys, obtaining the login request $(cid, acid, \pi_D, M)$ that the broker forwards.
Upon $(\mathsf{BrokerView}, sid, uid_{rp}, |info|)$, it computes $auid$ from the leaked $uid_{rp}$ and the blinding factor $k_1$ of the simulated RP, exactly as $\mathcal{S}_{\mathsf{RP}}$ does in Case 1. It also replays the leaked $uid_{rp}$ as the delivery message from the simulated RP to the broker.
Every such message is therefore a function of $\mathcal{L}_{\mathsf{IdB}}$ and randomness sampled with the same distribution as in the real protocol, and none of them depends on the user identity $uid$ or the contents of $info$.
Hence, $H_0$ and $H_1$ are perfectly indistinguishable.

\re{\textbf{Hybrid} $\textbf{H}_2$ is the same as $H_1$, except that the real session key $K_S$ derived via $\mathsf{AnonyAKA}$ is replaced with a uniformly random key $K_{R}$. This step rests on the key indistinguishability of the $\mathsf{X3DH}$ key agreement rather than on any property of the encryption scheme. In $\mathsf{AnonyAKA}$, $K_S = \mathsf{KDF}(dh_1\|dh_2\|dh_3)$ is derived from secret keys held only by the RP and the IdP, and the broker observes only the corresponding public keys. Any distinguisher between $H_1$ and $H_2$ therefore yields a PPT adversary $\mathcal{B}$ that breaks DDH in the random oracle model, where $\mathcal{B}$ embeds its challenge into the ephemeral Diffie--Hellman terms and answers $\mathsf{KDF}$ queries with the random oracle.
Hence, $H_1$ and $H_2$ are computationally indistinguishable.}

\textbf{Hybrid} $\textbf{H}_3$ is the same as $H_2$, except that the ciphertext $\mathsf{info}_{Enc} = \mathsf{Enc}_{K_{R}}(Info)$ is replaced with $\mathsf{Enc}_{K_{R}}(0^{|Info|})$. Since $K_R$ is now a uniformly random key that is independent of the broker's view, the IND-CPA security of the symmetric encryption scheme applies, and the broker cannot distinguish the encryption of real data from the encryption of zeros. Note that only $|Info|$ is preserved across this hybrid, which is exactly the leakage that $\mathsf{BrokerView}$ grants the broker in Figure~\ref{fig:ideal-orionlink}.
Hence, $H_2$ and $H_3$ are computationally indistinguishable.

Since $H_0 \approx_c H_3$, the real execution with $\mathcal{A}_{\mathsf{IdB}}$ is computationally indistinguishable from the ideal execution with $\mathcal{S}_{\mathsf{IdB}}$ and $\mathcal{F}_{\mathsf{OrionLink}}$.
Moreover, the inputs that $\mathcal{S}_{\mathsf{IdB}}$ receives from $\mathcal{F}_{\mathsf{OrionLink}}$ are exactly the leakage $\mathcal{L}_{\mathsf{IdB}}$, so this case establishes the broker obliviousness stated in Definition~\ref{def:broker-obliviousness}.

\re{
\subsubsection*{Case 4: Joint Corruption of the Broker and RPs ($\mathcal{A}_{\mathsf{Joint}}$)}
Ensure that a semi-honest broker that pools its view with an arbitrary set of malicious colluding RPs still cannot link a user's identity across RPs.

\textbf{Simulator $\mathcal{S}_{\mathsf{Joint}}$.}
$\mathcal{S}_{\mathsf{Joint}}$ runs $\mathcal{A}_{\mathsf{Joint}}$ internally and simulates the honest IdP, the honest RPs, and the users. Toward the corrupted RPs it behaves exactly as $\mathcal{S}_{\mathsf{RP}}$ in Case~1, verifying every presentation and answering the $\mathsf{Judge}$ queries of $\mathcal{F}_{\mathsf{OrionLink}}$. Toward the broker it behaves exactly as $\mathcal{S}_{\mathsf{IdB}}$ in Case~3, producing the broker-visible messages of honest-RP sessions from the leakage $\mathcal{L}_{\mathsf{IdB}}$ and fresh randomness.

\textbf{Hybrid} $\textbf{H}_0$ The real world protocol execution.

\textbf{Hybrid} $\textbf{H}_1$ is the same as $H_0$, except that every message in the joint view other than the pseudonyms and the ciphertexts $\mathsf{info}_{Enc}$ of honest-RP sessions is produced by $\mathcal{S}_{\mathsf{Joint}}$. The broker-side messages are generated from $\mathcal{L}_{\mathsf{IdB}}$ exactly as in Hybrid $H_1$ of Case~3, and the messages exchanged with the corrupted RPs are generated exactly as in Hybrids $H_1$ through $H_3$ of Case~1. Hence $H_0$ and $H_1$ are indistinguishable by the arguments of those hybrids.

\textbf{Hybrid} $\textbf{H}_2$ is the same as $H_1$, except that in honest-RP sessions the session key is replaced with a uniformly random key and $\mathsf{info}_{Enc}$ with an encryption of zeros. The coalition holds no key material of these sessions, so the same reductions to DDH in the random oracle model and to IND-CPA apply, and $H_1$ and $H_2$ are computationally indistinguishable.

\textbf{Hybrid} $\textbf{H}_3$ is the same as $H_2$, except that every pseudonym the coalition obtains, whether unblinded by a corrupted RP or delivered to the broker in an honest-RP session, is replaced by the uniformly random $uid_{rp}$ that $\mathcal{F}_{\mathsf{OrionLink}}$ samples per user--RP pair, embedded consistently as in Case~1. Distinguishing $H_2$ from $H_3$ means distinguishing the family $f_{uid}(cid) = cid^{\,uid}$ from an independent random function on the observed user--RP pairs, and a standard hybrid over these pairs reduces each step to DDH in $\mathbb{G}_1$. The reduction embeds a challenge $(g, g^{\alpha}, g^{\beta}, Z)$ by implicitly setting the identifier of the target user to $\alpha$ and the $cid$ of the target RP to $g^{\beta}$, samples every other $cid$ as $g^{r}$ with a known exponent so that all remaining pseudonyms stay computable, and places $Z$ at the target pair. Hence $H_2$ and $H_3$ are computationally indistinguishable.

Since $H_0 \approx_c H_3$, the real execution with $\mathcal{A}_{\mathsf{Joint}}$ is computationally indistinguishable from the ideal execution with $\mathcal{S}_{\mathsf{Joint}}$ and $\mathcal{F}_{\mathsf{OrionLink}}$.
}

\noindent \textbf{Conclusion.}
Combining these cases, we have proved that $\Pi_{\mathsf{OrionLink}}$ satisfies the defined security goals.
\re{
\subsection{Formal Privacy Definitions and Derivations}
We now state the three privacy definitions summarized in Section~\ref{sec:security-analysis} and establish each of them under the assumptions of Theorem~\ref{thm:security}.

\begin{definition}[Broker Obliviousness]\label{def:broker-obliviousness}
Let $\mathsf{view}^{\Pi}_{\mathsf{IdB}}$ denote the view of an honest-but-curious broker in an execution of a brokered SSO protocol $\Pi$, and let $\mathcal{L}_{\mathsf{IdB}}$ denote the leakage granted by $\mathsf{BrokerView}$ in Figure~\ref{fig:ideal-orionlink}, namely $(cid, D)$ at RP registration, $(sid, cid, acid)$ at login initialization, and $(sid, uid_{rp}, |info|)$ at login finalization. $\Pi$ achieves \emph{broker obliviousness} if there exists a PPT simulator $\mathcal{S}_{\mathsf{IdB}}$ such that $\mathcal{S}_{\mathsf{IdB}}(1^{\lambda}, \mathcal{L}_{\mathsf{IdB}})$ is computationally indistinguishable from $\mathsf{view}^{\Pi}_{\mathsf{IdB}}$.
\end{definition}

\noindent\textit{Derivation.} This definition is exactly the statement established by Case~3 above, so no separate reduction is required.

\begin{definition}[Untraceability]\label{def:untraceability}
A brokered SSO protocol $\Pi$ achieves \emph{untraceability} if for any $\lambda$ and any PPT adversary $\mathcal{A}$ there exists a negligible function $\negl(\cdot)$ such that
\[
\mathsf{Adv}^{\mathsf{untr}}_{\mathcal{A}} :=
\left|\Pr[\mathsf{Exp}^{\mathsf{untr}}_{\Pi}(\mathcal{A},\lambda)=1] - \tfrac{1}{2}\right| \le \negl(\lambda),
\]
where the experiment $\mathsf{Exp}^{\mathsf{untr}}_{\Pi}(\mathcal{A},\lambda)$ is defined as follows.
\begin{itemize}[leftmargin=*]
\item The challenger runs the honest key generation to obtain the issuance key pair $(\mathsf{isk}_c, \mathsf{ivk}_c)$ and hands both keys to $\mathcal{A}$, who plays the IdP.
\item The challenger plays two honest RPs with distinct identity descriptors $D_0 \neq D_1$ chosen by $\mathcal{A}$ and obtains their credentials $(\mathcal{DC}_0, \mathcal{SC}_0)$ and $(\mathcal{DC}_1, \mathcal{SC}_1)$ from $\mathcal{A}$ via the registration protocol.
\item On each query $d \in \{0,1\}$ of $\mathcal{A}$, the challenger returns a fresh login request $(acid, M)$ generated from $(\mathcal{DC}_d, \mathcal{SC}_d)$.
\item The challenger samples $b \narrowxleftarrow{\$} \{0,1\}$ and outputs a challenge request $(acid^{*}, M^{*})$ generated from $(\mathcal{DC}_b, \mathcal{SC}_b)$.
\item $\mathcal{A}$ outputs a guess $b'$, and the experiment outputs $1$ iff $b' = b$.
\end{itemize}
\end{definition}

\noindent\textit{Derivation.} In the experiment the adversary plays the IdP acting as both the credential issuer and the credential verifier, so issuer--verifier collusion is within scope. The environment corrupts the IdP semi-honestly, generates the issuance keys honestly, and hands them to $\mathcal{A}$, exactly as in the experiment. In the ideal execution the IdP observes only $(\mathsf{LoginReq}, sid, acid)$, where $\mathcal{F}_{\mathsf{OrionLink}}$ samples $acid$ fresh at random for honest RPs independently of $D_0$ and $D_1$, so the advantage is zero. The transition to the real execution is bounded by Hybrids $H_1$ and $H_2$ of Case~2. By Theorem~\ref{thm:security}, the advantage of any PPT adversary in the real execution is therefore negligible. We in fact cover an IdP that generates $(\mathsf{isk}_c, \mathsf{ivk}_c)$ maliciously rather than honestly, because the unlinkability assumption on $\mathsf{AC}$ used above already permits the adversary to generate the key pair itself (Section~\ref{sec:Security Properties of Anonymous Credential}), although such an IdP lies beyond the honest-but-curious scope of our threat model.

\begin{definition}[Unlinkability]\label{def:unlinkability}
A brokered SSO protocol $\Pi$ achieves \emph{unlinkability} if for any $\lambda$ and any PPT adversary $\mathcal{A}$ there exists a negligible function $\negl(\cdot)$ such that
\[
\mathsf{Adv}^{\mathsf{unlink}}_{\mathcal{A}} :=
\left|\Pr[\mathsf{Exp}^{\mathsf{unlink}}_{\Pi}(\mathcal{A},\lambda)=1] - \tfrac{1}{2}\right| \le \negl(\lambda),
\]
where the experiment $\mathsf{Exp}^{\mathsf{unlink}}_{\Pi}(\mathcal{A},\lambda)$ is defined as follows.
\begin{itemize}[leftmargin=*]
\item $\mathcal{A}$ controls the broker and an arbitrary set of colluding RPs, while the challenger plays the IdP and the honest users.
\item The challenger creates users with secret identifiers $uid_i \narrowxleftarrow{\$} \mathbb{Z}^*_p$ and, on each RP registration of $\mathcal{A}$, returns a fresh identifier $cid \narrowxleftarrow{\$} \mathbb{G}^*_1$.
\item $\mathcal{A}$ adaptively queries an oracle $\mathcal{O}_{\mathsf{pseu}}(i, cid)$ that returns the pseudonym $cid^{uid_i}$ observed in a login of user $i$ at the RP identified by $cid$.
\item $\mathcal{A}$ outputs two registered identifiers $cid_1 \neq cid_2$ and two user indices $i_0 \neq i_1$.
\item The challenger samples $b \narrowxleftarrow{\$} \{0,1\}$ and returns the challenge pair $(cid_1^{uid_{i_0}}, cid_2^{uid_{i_b}})$.
\item $\mathcal{A}$ outputs a guess $b'$, and the experiment outputs $1$ iff $b' = b$ and no query to $\mathcal{O}_{\mathsf{pseu}}$ involved $cid_2$ with index $i_0$ or $i_1$.
\end{itemize}
\end{definition}

\noindent\textit{Derivation.} The adversary in this experiment, a broker that pools its view with colluding RPs, matches exactly the corruption pattern of Case~4, and every pseudonym in this joint view, including those the broker observes in honest-RP sessions, is an output of $\mathcal{O}_{\mathsf{pseu}}$. In the ideal execution $\mathcal{F}_{\mathsf{OrionLink}}$ samples every pseudonym independently and uniformly at random, so the advantage is zero. The transition to the real execution is bounded by Hybrid $H_3$ of Case~4, which reduces to the pseudorandomness of the DH-based OPRF $f_{uid}(cid) = cid^{\,uid}$ and hence to DDH in $\mathbb{G}_1$. By Theorem~\ref{thm:security}, the advantage of any PPT adversary in the real execution is therefore negligible.

The three derivations above yield the following corollary.

\begin{corollary}\label{cor:privacy}
Under the assumptions of Theorem~\ref{thm:security}, OrionLink achieves broker obliviousness, untraceability, and unlinkability in the sense of Definitions~\ref{def:broker-obliviousness}, \ref{def:untraceability}, and \ref{def:unlinkability}, respectively.
\end{corollary}
}
\subsection{Security of Forward-Secure Token Refresh}
We provide a brief proof that our forward-secure token refresh mechanism achieves both forward secrecy and post-compromise security.

\begin{definition}[Forward Secrecy and Post-Compromise Security] Suppose $\mathcal{A}$ obtains the key for any round of the session, messages sent before or after that round will not be compromised.
\end{definition}

\textit{Proof Sketch.}
We utilize a symmetric-key ratchet to achieve \textit{forward secrecy}. Each new session key used by the IdP to encrypt the refreshed access token is derived from the previous one using a key derivation function. And we utilize a Diffie–Hellman ratchet to achieve post-compromise security. The RP pre-generates a list of ephemeral DH key pairs and shares the public keys with the broker. When a token refresh is required, the broker forwards the next public key to the IdP, who performs a new DH exchange to inject fresh randomness into the session key derivation.

\clearpage
\section{Artifact Appendix}
\label{page:artifact-start}

This appendix describes how to reproduce the results in the main body of the
paper and is written to be read on its own. The artifact is a complete
prototype of \textit{OrionLink}, a privacy-preserving single sign-on (SSO)
system. In brokered SSO a user logs in to a website, the relying party (RP),
through an intermediary, the broker, which forwards the login to the identity
provider (IdP) that authenticates the user and issues the tokens. In
\textit{OrionLink} the broker is \emph{oblivious}: it routes logins without
learning the user's cross-RP identity or any personal data. The artifact
contains the four deployable roles of the paper's Implementation section
(Section~\ref{sec:Implementation}): IdP, broker, RP, and the Trusted
Computation Agent (TCA), a browser-side script on its own origin that performs
the user's client-side cryptography. It implements the four protocol phases
(Registration, Login Request, User Authentication and Consent, Login Response),
the two extensions (private authorization management, forward-secure token
refresh), and one measurement harness per experiment of the Evaluation section
(Section~\ref{sec:Evaluation}); about 9K lines of Go and JavaScript in
total.
``README~\S$x$'' below refers to \path{README.md} at the repository root.

\subsection{Description \& Requirements}

\subsubsection{How to access}
The artifact repository is available via a persistent DOI:
\url{https://doi.org/10.5281/zenodo.22815552}, and on GitHub:
\url{https://github.com/Ooo0ption/Orionlink}. The four roles are also
published as one prebuilt Docker image,
\path{walnutd/orionlink-ae:ndss2027-v2}.

\subsubsection{Hardware dependencies}
No special hardware. A commodity x86-64 machine with 8 cores, 16\,GB of RAM
and about 10\,GB of free disk runs every experiment and is the basis for the
time estimates below. The paper's results were produced on an AMD Ryzen~9
9950X3D (16 cores) with 96\,GB of RAM under Ubuntu~24.04.

\subsubsection{Software dependencies}
Two layers. (i)~To deploy the prototype and to run E3 the host needs only
Docker~$\ge$~24 with the Compose~v2 plugin, on Linux (Ubuntu~24.04
recommended) or on Windows with Docker Desktop in Linux-container mode; every
command below has a PowerShell equivalent in the README. The image is
self-contained and resolves no package from the network at run time.
(ii)~E1 and E2 are host-side harnesses that drive the containers from outside,
so the host additionally needs Go~$\ge$~1.25 (E1 and the unit tests),
Node.js~$\ge$~18 (E2, which installs its own Chromium) and Python~$\ge$~3.8
(E2, standard library only). Install commands for Ubuntu, macOS and Windows,
with version checks, are in README~\S3.2; \path{requirements.txt} is the
authoritative list. Nothing is ever installed inside a container.

\subsection{Artifact Installation \& Configuration}

README~\S3.1 gives the whole reproduction, from checkout to the three result
sets, as one numbered command sequence; the steps here and the experiments
below expand its rows. Every command is typed in the host shell at the
repository root.

\begin{enumerate}\itemsep1pt\parskip0pt
\item \textbf{Pull the image} (README~\S2.2). Set
\path{ORION_IMAGE=walnutd/orionlink-ae:ndss2027-v2} in the shell and run
\texttt{docker image pull}; the README shows the expected status line. To
build the same image from source instead, online or fully offline,
follow README~\S5 and continue here.
\item \textbf{Start the containers} (README~\S2.3).
\texttt{docker compose up -d -{}-wait} starts the roles in dependency order and
returns once all four report healthy; ports 3000--3003 must be free. The
README lists the four addresses and marks the only one to open by hand, the
RP at \path{http://localhost:3002/ssso/}.
\item \textbf{Install the host-side tools} (README~\S3.2):
\texttt{go mod download}, then \texttt{npm ci} and \texttt{npx playwright
install chromium} in \path{benchmark/e2/browser}, each once.
\end{enumerate}

The stack runs in one of two profiles, selected by \path{ORION_PROFILE}. The
default, \texttt{demo}, is the deployment to click through by hand.
\texttt{bench} runs the same binary and pages but registers the RP at startup,
seeds synthetic users and mounts the load-test endpoints; E2 and E3 require it.

\textbf{Optionally}, the front end of the prototype lets the evaluator walk
through the \textit{OrionLink} flow one button at a time (README~\S4),
registering the RP, signing in as the demo user, refreshing the token and
revoking the authorization. This walk-through is \emph{not} a prerequisite
for any experiment.

\subsection{Major Claims}

The major claims of the paper that this artifact substantiates are summarized as
follows.

\begin{itemize}
\item (C1): \textit{OrionLink}'s privacy mechanisms incur only marginal local
computation, and the anonymous authenticated key agreement $\mathsf{AnonyAKA}$
(AAKA), which RP and IdP run through the broker, is the most expensive stage on
the login path. This claim is substantiated by Experiment (E1), which reproduces
Table~\ref{tab:local-overhead}.
\item (C2): End-to-end browser login remains within interactive bounds and
degrades gracefully with network delay, growing approximately linearly with the
network round-trip time (RTT) for both \textit{OrionLink} and the standard brokered-SSO baseline. This
claim is substantiated by Experiment (E2), which reproduces
Table~\ref{tab:latency-breakdown} and Fig.~\ref{fig:user login latency}.
\item (C3): The system stays low-latency under sustained load on its most
expensive interface, the $\mathsf{AnonyAKA}$ path jointly executed by RP, broker
and IdP. This claim is substantiated by Experiment (E3), which reproduces
Fig.~\ref{fig:plot2}.
\end{itemize}

\subsection{Evaluation}

One experiment per claim. Each is driven by a single runner under
\path{benchmark/} and writes its report to \path{benchmark/results/};
iteration counts and arrival rates are arguments, so an experiment runs at the
paper's setting or at a reduced one. All commands are typed in the host shell
at the repository root, and the evaluator never has to enter a container. The
two steps that do run inside one, the k6 load generator of E3 and the netem
delay setting, are started from the host by Compose.

\medskip
\noindent\textbf{Experiment (E1)} \emph{Local computation overhead}
[Table~\ref{tab:local-overhead}] ($\le$10 compute-min; README~\S3.3). Supports
claim C1. It measures each privacy-enhancing stage in isolation and prints the
six rows of Table~\ref{tab:local-overhead}.

\textit{[Configuration Parameters]} The paper's setting is 1{,}000 runs per
stage; the argument defaults to 100. Every stage is measured in-process.

\textit{[Preparation and Execution]} Needs Go on the host (step~3). The stack
must be \emph{down} (\texttt{docker compose down}): it would compete for CPU,
and the unit tests bind port 3000 themselves.

\begin{footnotesize}
\begin{verbatim}
go test ./benchmark/unit/
./benchmark/run-e1.sh 1000
\end{verbatim}
\end{footnotesize}

\textit{[Results]} The mean time per operation for the six stages, plus their
total. The overhead stays small at every stage, which supports the claim that
the privacy mechanisms of \textit{OrionLink} add negligible local computation.
Absolute values scale with single-core speed, so a deviation of up to a factor
of two is expected on slower hardware.

\medskip
\noindent\textbf{Experiment (E2)} \emph{End-to-end latency}
[Table~\ref{tab:latency-breakdown}, Fig.~\ref{fig:user login latency}]
($\le$45 compute-min; README~\S3.4). Supports claim C2. A Playwright
browser-automation harness drives the running containers and produces every
row of Table~\ref{tab:latency-breakdown}: the \textit{SSO Flow} rows from a
full login, the \textit{Authorization management} row from the PIN-based
decryption in the TCA, and the \textit{Token refresh} row as an authenticated
browser request that takes the live RP\,$\rightarrow$\,broker\,$\rightarrow$\,IdP
path and back.

\textit{[Configuration Parameters]} The paper's setting is 100 runs per
configuration, at emulated RTTs of 0, 40, 100 and 200\,ms.

\textit{[Preparation and Execution]} Needs Node.js and Python on the host
(step~3). Start the stack in the bench profile, then run the harness; repeat at
each RTT after applying the per-direction delay with
\texttt{./netem-compose.sh set-egress <delay>} (README~\S3.6).

\begin{footnotesize}
\begin{verbatim}
ORION_PROFILE=bench docker compose up -d --wait
ORION_PROFILE=bench ./benchmark/run-e2.sh 100
\end{verbatim}
\end{footnotesize}

\noindent The brokered OpenID Connect (OIDC) baseline that Table~\ref{tab:latency-breakdown} compares
against is a separate stack with its own image and ports, documented in
\path{benchmark/oidc-baseline/}.

\textit{[Results]} The runner ends by printing Table~\ref{tab:latency-breakdown}
(\path{benchmark/e2/summarize.py} re-renders it from the reports). Browser
timings are noisy, so accept a deviation of roughly 50\% on absolute values.
The claim to check is that the \textit{OrionLink} total at 40\,ms RTT is on the
order of two seconds and grows linearly with the RTT.

\medskip
\noindent\textbf{Experiment (E3)} \emph{Stress test of the
$\mathsf{AnonyAKA}$ interface} [Fig.~\ref{fig:plot2}] ($\le$20 compute-min;
README~\S3.5). Supports claim C3. It drives the AAKA work of the three roles
under a rising arrival rate and produces the three lines of
Fig.~\ref{fig:plot2}, one per role.

\textit{[Configuration Parameters]} The paper's setting is an offered load of
100, 500, 1{,}000, 1{,}300 and 1{,}500 requests per second (RPS), generated by
k6 with \texttt{constant-arrival-rate}, 60\,s per rate, for each of the three
roles. By default the runner starts a lightweight configuration, 10 and
200\,RPS, that any machine sustains; the paper's ladder is requested with
\texttt{-{}-rps}.

\textit{[Preparation and Execution]} Needs only Docker. Start the stack in the
bench profile with \path{ORION_NETEM_DELAY} empty, then start the load
generator, a fifth Compose service that carries the pinned k6 binary and the
E3 scripts inside the submitted image. The command is typed on the host;
Compose creates and removes the container.

\begin{footnotesize}
\begin{verbatim}
docker compose --profile benchmark run --rm e3
docker compose --profile benchmark run --rm e3 \
    --rps 100,500,1000,1300,1500
\end{verbatim}
\end{footnotesize}

\noindent The first line is the lightweight default (about eight minutes;
\texttt{-{}-quick} shortens it to two); the second reproduces the paper's
ladder. Anything after \texttt{e3} is passed to \path{benchmark/run-e3.sh}.

\textit{[Results]} One JSON report per role with, for each rate, the target
and achieved rate, the p50/p95/p99 latency, the server-side p95 and the
dropped-request count; a rate is a valid data point only when nothing was
dropped and the achieved rate matches the target. Plot p95 against the rate to
obtain Fig.~\ref{fig:plot2}. The result to check is the paper's: at
1{,}000\,RPS the p95 latency stays at about 10\,ms for the IdP and below that
for RP and broker, i.e., flat across the first three rates. Above 1{,}000\,RPS
the IdP line rises steeply as its pairing-bound AAKA work saturates the CPU;
on a machine with fewer cores than the paper's this knee moves to a lower
rate, which does not affect the claim.

\subsection{Customization}

The deployment is configured through \path{.env.example}: \path{ORION_IMAGE}
selects the image, \path{ORION_PROFILE} the profile,
\path{ORION_NETEM_DELAY} the egress delay, and \path{ORION_*_URL} /
\path{ORION_*_PORT} relocate the roles. \path{benchmark/README.md} documents
every remaining parameter and the troubleshooting steps.

\subsection{Notes}

The prototype is dominated by public-key cryptography and is bound to the
CPU, so every experiment is sensitive to other load on the machine and the
tolerances quoted above need not cover every deployment; close other workloads
first. In (E3) the load generator shares the cores with the server roles and
every request crosses the Docker network stack, so running the roles and k6
directly on the host sustains a higher rate.

\end{document}